%% file: Tractor_Action_For_Einstein_Gravity.tex
\documentclass[10pt]{amsart}
\usepackage[english]{babel}

 \usepackage[a4paper,layout=a4paper,centering, left=2.4cm, right=2.4cm, top =2.5cm, bottom= 2.5cm ]{geometry}

\usepackage{package-notes}
\usepackage{comment}
\usepackage{cite}
\usepackage{enumitem}

\usepackage{amssymb,amscd,amsmath}
\usepackage{bbold,bm}
\usepackage{color}
\usepackage{tikz}
\numberwithin{equation}{section}

\usepackage{mathrsfs}

\input{macros}

\excludecomment{comment1}
\excludecomment{comment4}

\excludecomment{DetailedCheckForTransformationRules}

\title[Einstein gravity as a gauge theory for the conformal group]{Einstein gravity as a gauge theory for the conformal group
}
\author{ Yannick Herfray}
\address[Y.~Herfray]{Département de Mathématique\\
	Université Libre de Bruxelles}
\email{Yannick.Herfray@ulb.ac.be}
\author{Carlos Scarinci}
\address[C.~Scarinci]{Institute of Mathematical Sciences \\ Ewha Womans University
}
\email{cscarinci@ewha.ac.kr}

\begin{document}

\begin{abstract}
General Relativity in dimension $n = p + q$ can be formulated as a gauge theory for the conformal group $\SO\left(p+1,q+1\right)$, along with an additional field reducing the structure group down to the Poincaré group $\ISO\left(p,q\right)$. In this paper, we propose a new variational principle for Einstein geometry which realizes this fact. Importantly, as opposed to previous treatments in the literature, our action functional gives first order field equations and does not require supplementary constraints on gauge fields, such as torsion-freeness.

Our approach is based on the ``first order formulation'' of conformal tractor geometry. Accordingly, it provides a straightforward variational derivation of the tractor version of the Einstein equation. To achieve this, we review the standard theory of tractor geometry with a gauge theory perspective, defining the tractor bundle a priori in terms of an abstract principal bundle and providing an identification with the standard conformal tractor bundle via a dynamical soldering form. This can also be seen as a generalization of the so called Cartan-Palatini formulation of General Relativity in which the ``internal'' orthogonal group $\SO\left(p,q\right)$ is extended to an appropriate parabolic subgroup $P\subset\SO\left(p+1,q+1\right)$ of the conformal group.
\end{abstract}
\maketitle

\section{Introduction}

\subsection{Conformally compact manifolds and the Cartan-Palatini functional}\mbox{}

The variational approach to General Relativity can be given a first order formulation via the Cartan-Palatini action functional, see e.g. \cite{krasnov_formulations_2020},
\begin{equation}\label{Introduction Cartan-Palatini}
	S_{CP}\left[\tilde\omega,\et\right] = \int_{\Mt} \left(\tfrac{1}{2(n-2)!} F^{\tilde\omega}{}^{ij} - \tfrac{1}{n!}\;\Lambda\;\et^i\wedge \et^j \right)\wedge \epsilon_{ijk_1\cdots k_{n-2}}\et^{k_{1}}\wedge \cdots\wedge \et^{k_{n-2}}.
\end{equation}
Here, the soldering form $\et^i \in \Gamma(T^* \Mt \otimes \cE)$ is a bundle isomorphism from the tangent space $T\Mt$ to a so-called ``internal'' tangent bundle $\cE\to\Mt$, equipped with a canonical fiber metric $\eta_{ij}$ of signature $(p,q)$, while the spin-connection $\nabla^{\tilde\omega}$ is an independent metric-compatible connection on $\cE$.

A pseudo-Riemannian metric $\gt := \et^i \et^i \eta_{ij}$ is then defined on the $n$-dimensional space-time $\Mt$ via pull-back of the fiber metric $\eta_{ij}$, and the field equation obtained by varying the connection --- $d^{\tilde\omega} \et^i=0$ --- requires the pull-back of $\nabla^{\tilde\omega}$ to coincide with the Levi-Civita connection $\nabla^{\gt}$. Evaluating \eqref{Introduction Cartan-Palatini} on the solution of this field equation, one immediately recovers the usual Einstein-Hilbert action for $\gt$. In particular, the field equation obtained by varying soldering form gives a condition on the curvature of $\nabla^{\tilde\omega}$, which is equivalent to the Einstein equations
\begin{equation*}
	\Ric^{\gt}-\frac{1}{2}R^{\gt} \gt+\Lambda \gt=0.
\end{equation*}

The story is more subtle for manifolds with boundary. The variational principle must be complemented by appropriate boundary conditions and the functional \eqref{Introduction Cartan-Palatini} by additional boundary terms, ensuring well-definiteness of the functional or at least of its variations. An important example, particularly motivating for the present work, is that of \emph{conformally compact} boundary conditions \cite{penrose_conformal_2011,frauendiener_conformal_2004}. Starting with a compact manifold $M$ with boundary $\partial M$ and a ``physical'' soldering form $\et^i$ on its interior $\Mt$, one introduces a decomposition $\et^i=\sigma^{-1} e^i$ in terms of an auxiliary ``unphysical'' soldering $e^i$ and a scalar field $\sigma$ on $M$. This results in two metrics $\gt = \et^i \et^j \eta_{ij}$ and $g= e^i e^j \eta_{ij}$, which are related by conformal rescaling 
\begin{equation*}
\gt = \sigma^{-2} g.
\end{equation*}
The fields $\left(e^i, \sigma\right)$ are then subject to simple boundary conditions : $\sigma$ is required to identically vanish at $\partial M$, with nowhere vanishing derivative, while $e^i$ is required to extend the soldering from $\sigma\et^i$ with a prescribed level of differentiability at $\partial M$, see \cite{friedrich_peeling_2018} for a discussion on the physical relevance of this last criteria.

Since these important boundary conditions are imposed on the unphysical fields, one is lead to consider variational principles for Einstein metrics written explicitly in terms of $\left(\nabla^\omega , e^i , \sigma \right)$. In what follows, we consider the \emph{Cartan-Palatini-Weyl functional}
\begin{align}\label{Introduction Cartan-Palatini-Weyl}
S_{CPW}\left[\omega,e,\sigma\right] = \int_M \sigma^{2-n}\bigg(\tfrac{1}{2}F^\omega{}^{ij}
+d^{\omega}(\Upsilon^i e^j)-\left( \tfrac{n-2}{2n}|\Upsilon|^2
+\sigma^{-2}\tfrac{(n-2)!}{n!} \Lambda\right) e^i\wedge e^j\bigg)\wedge(\star  e^{n-2})_{ij},
\end{align}
where we use $\Upsilon_i e^i := \sigma^{-1}d\sigma$ and 
\begin{equation*}
\left(\star e^{n-k} \right)_{i_1\cdots i_k} :=\frac{1}{(n-k)!} \;\epsilon_{i_1\cdots i_k j_{1} \cdots j_{n-k}}\; e^{j_{1}}\wedge \cdots\wedge e^{j_{n-k}}.
\end{equation*} 

As compared to \eqref{Introduction Cartan-Palatini}, the functional \eqref{Introduction Cartan-Palatini-Weyl} depends on one additional scalar field $\sigma$. This is however compensated by an additional invariance under Weyl rescalings: for every positive scalar function $\lambda\in C^\infty(M,\R_+)$, \eqref{Introduction Cartan-Palatini-Weyl} is invariant under the replacements
\begin{align}\label{Introduction Cartan-Palatini-Weyl Invariance}
\omega{}^i{}_j &\mapsto \omega^i{}_j +\lambda\left(e^ir_j -r^ie_j\right), & e^i& \mapsto \lambda e^i, &\sigma & \mapsto \lambda \sigma,
\end{align}
with $r_ie^i:= \lambda^{-2}d\lambda$. In particular, away from $\partial M$, we can take $\sigma=1$ as a gauge-fixing condition and recover the original Cartan-Palatini functional.

\begin{DetailedCheckForTransformationRules}	
	
	\XDnote{One can take the following as a compatibility check:
	\begin{align*}
	d^{\hat\omega}\hat e^i&=d\hat e^i+\hat\omega^i{}_j\wedge \hat{e}^j
	\cr&
	=d(\lambda e^i)+ \lambda (\omega^i{}_j+\lambda\left(e^i r_j-r^ie_j\right))\wedge e^j
	\cr&
	=\lambda d^\omega e^i+\lambda^2\Big(\lambda^{-2}d\lambda\wedge e^i+e^i r_j\wedge e^j\Big)
	\cr&
	=\lambda d^\omega e^i+\lambda^2\Big(\lambda^{-2}d\lambda-r_j e^j\Big)\wedge e^i
	\end{align*}
	In other words, if
	\begin{align*}
	\hat\omega^i{}_j=\omega^i{}_j+\lambda\left(e^i r_j-r^ie_j\right),
	\qquad
	r_j e^j= \lambda^{-2}d\lambda.
	\end{align*}	
	
	then $d^{\hat\omega}\hat e^i=  \lambda d^{\omega} e^i$.
}
\end{DetailedCheckForTransformationRules}

More directly, we can easily see that \eqref{Introduction Cartan-Palatini-Weyl} is indeed a functional for Einstein metrics by considering its variation with respect to the connection $\omega$, solving the resulting field equations $d^\omega e^i=0$, and then evaluating the functional on the corresponding solution. This results in the \emph{Einstein-Hilbert-Weyl functional}
\begin{equation*}
	S_{EHW}[ g,\sigma] := S_{EH}[ \sigma^{-2} g ]
	=\int_M\sigma^{2-n}\Big( \tfrac{1}{2}R^g  + (n-1)(\nabla^g\cdot\Upsilon
	-\tfrac{n-2}{2}|\Upsilon|^2) - \sigma^{-2}\Lambda \Big)dv_g,
\end{equation*}
where now $\Upsilon := \sigma^{-1}d\sigma$, or equivalently, introducing $\phi := \sigma^{\frac{2-n}{2}}$ in the more familiar functional for a conformally coupled scalar field:
\begin{equation}\label{Introduction Einstein-Hilbert-Weyl}
	S[ g,\phi] := S_{EH}[ \phi^{\frac{4}{n-2}} g ]
	=\int_M  \Big(\tfrac{1}{2}\phi^2R^g   -2\tfrac{n-1}{n-2}\phi \Delta^g\phi -  \phi^{\frac{2n}{n-2}}\Lambda \Big)dv_g.
\end{equation}
Once again, the classical Einstein-Hilbert functional can be recovered through a partial gauge fixing away form $\partial M$. 
\begin{comment1}
\XDnote{
	For any $a\in\R^\times$
	$$\nabla^g\cdot\Upsilon-\tfrac{(n-2)}{2}|\Upsilon|^2=\sigma^{-1}\Delta^g\sigma-\tfrac{n}{2}|\Upsilon|^2=\tfrac{1}{a}\sigma^{-a}\Delta^g\sigma^a-(\tfrac{n-2}{2}+a)|\Upsilon|^2=-\tfrac{2}{n-2}\sigma^{-\frac{2-n}{2}}\Delta^g\sigma^{\frac{2-n}{2}}$$
}
\end{comment1}

The functional \eqref{Introduction Einstein-Hilbert-Weyl} goes back at least to \cite{deser_scale_1970,zumino_effective_1971} and has been studied extensively in the physics literature, especially generalizations obtained by coupling this conformal scalar field to the usual Einstein-Hilbert action (see \cite{bekenstein_exact_1974,halliwell_conformal_1989,xanthopoulos_einstein_1992,abreu_exact_1994} and references therein), see also \cite{padmanabhan_conformal_1985} for other types of modifications of this Lagrangian preserving invariance under Weyl rescalings. Super-symmetric extensions have also been investigated in \cite{kaku_poincare-supergravity_1978} and Part 5 of \cite{freedman_supergravity_2012}.

Despite being a crucial feature of \eqref{Introduction Cartan-Palatini-Weyl} and \eqref{Introduction Einstein-Hilbert-Weyl}, the invariance under Weyl rescaling is not manifest in either functional. And, indeed, it is a generic and unpleasant feature of (the naive approach to) conformal geometry that Weyl rescaling is rather difficult to implement. In the study of conformally compact manifolds this leads to a dilemma well-known to practitioners: one either works with the ``physical'' field $\gt = \sigma^{-2} g$, having then a direct geometric interpretation for the field equations but with little control over the boundary behavior, or one works with ``unphysical'' fields $g$ and $\sigma$, with a better geometric interpretation for the boundary conditions but with more complicated field equations and having also to deal with cumbersome Weyl transformations.

Such dilemma can be resolved using more refined tools from conformal geometry. Namely, a reformulation of the conformally compact boundary conditions is available in terms of \emph{conformal tractor calculus} \cite{gover_almost_2010,curry_introduction_2018}. With this, all fields and their field equations 1) have a well-defined behavior over the whole manifold $M$, including the boundary $\partial M$, 2) have natural geometric interpretations and, most importantly, 3) are manifestly invariant under Weyl transformations. The resulting picture is conceptually compelling and was successfully used in \cite{gover_conformal_2007,gover_boundary_2014,gover_poincare-einstein_2015} to produce new (both local and global) differential operators on the boundary of Poincaré--Einstein manifolds and in \cite{herfray_asymptotic_2020,herfray_tractor_2021} to invariantly encode the geometry of gravitational radiation at null-infinity. We refer the reader to \cite{bailey_thomass_1994,curry_introduction_2018} for a comprehensive introduction to the tractor formalism in conformal geometry.

Ultimately, the functional \eqref{Introduction Cartan-Palatini-Weyl} raises two main difficulties which are fundamentally related: on the one hand Weyl invariance is not manifest, on the other hand the geometrical meaning of the field equations is obscure. In the present article we will make use of tractor calculus to construct \emph{tractor functionals} for General Relativity which, at the same time, are equivalent to \eqref{Introduction Cartan-Palatini-Weyl}, are manifestly Weyl invariant and have first order field equations with a direct geometric interpretations (in the context of conformal geometry).

It should be noted here that tractor functionals have already made appearances in the literature \cite{gover_tractors_2009,gover_weyl_2009,bonezzi_gravity_2010}. The first two reference were essentially concerned with field theories on a fixed conformal background, although one can also find in these references a tractor functional for Einstein metrics. The third reference discuss this functional and its relationship to six dimensions.  We would like to point out, however, that in these works the tractor fields are not fundamental but parametrized in terms of other variables (the scale tractor is a function of the conformal metric and a scale). As such, these tractor fields cannot be varied independently, and the result is no more than a rewriting of the Einstein-Hilbert-Weyl functional \eqref{Introduction Einstein-Hilbert-Weyl} in tractor language. As a consequence, the resulting field equations in \cite{gover_tractors_2009,gover_weyl_2009,bonezzi_gravity_2010} are of second order.

The presence of implicit constraints and the high order character of the field equations seems to be, to the best of our knowledge, a common feature of all previous approaches to gravity based on the conformal group.  In fact, the functionals presented in \cite{ferrara_unified_1977,kaku_gauge_1977,kaku_superconformal_1977,kaku_properties_1978,kaku_poincare-supergravity_1978,townsend_simplifications_1979} or \cite{merkulov_supertwistor_1985} --- none of which apply to General Relativity, but rather to Conformal Gravity --- and those in \cite{wheeler_weyl_2014,attard_weyl_2016,wheeler_general_2019} --- which are functionals for General Relativity --- all share this unpleasant trait. Typically, the torsion-free condition on the connection, even thought necessary to obtain General Relativity, does not follow from the field equations but is added in by hand. In this sense, these constructions do not achieve genuine functionals for gauge connections valued in the conformal group. It is also worth mentioning \cite{preitschopf_conformal_1999} where tractor like Lagrangians where constructed in terms a $(n+2)$-dimensional ambient manifold. The problems discussed above are then avoided at the price of working on higher dimensional manifolds.

\subsection{A lightening overview of the first order tractor functional for General Relativity}\mbox{}

Let us now give a brief description of our tractor functionals. In following subsection we will provide a more mathematical account of its relations to the geometry of the tractor bundle.

Here the discussion will follow a more physical spirit, along the lines of \cite{ferrara_unified_1977,kaku_gauge_1977,kaku_superconformal_1977,kaku_properties_1978,kaku_poincare-supergravity_1978}. Accordingly, we take the view that a \emph{tractor} is simply a field with values in the fundamental representation of $\SO(p+1,q+1)$, denoted by
\begin{align*}
	I^I = \begin{pmatrix} \sigma \\ \mu^i \\ \rho \end{pmatrix}\in\R^{p+1,q+1},
	& &
	\sigma\in \R,\quad\mu^i\in \R^{p,q},\quad \rho \in \R.
\end{align*}
The ``internal vector indices'' $i,j,\cdots$ ranging between $1$ and $n = p+q$ and ``tractor indices'' $I,J,\cdots$ ranging between $1$ and $n+2$.

In the Cartan-Palatini formulation of General Relativity, the internal flat metric $\eta_{ij}$ is provided as a background structure, used to raise and lower internal vector indices, thus restricting the local symmetry group of the theory to the Lorentz group $\SO(p,q)$. In comparison, the tractor formulation also comes equipped with a flat tractor metric $h_{IJ}$
\begin{align*}
	h_{IJ} = \begin{pmatrix}  0 & 0 & 1 \\ 0 & \eta_{ij} & 0 \\ 1 & 0 & 0 \end{pmatrix},
	& &
	|I|^2=I_I I^I=2\sigma \rho+|\mu|^2,
\end{align*}
but, additionally, with a preferred null direction generated by the so called position tractor $X^I$
\begin{align*}
	X^I = \begin{pmatrix} 0 \\ 0 \\ 1 \end{pmatrix},
	& &
	|X|^2=X_I X^I=0.
\end{align*}
This restricts the local symmetry group from $\SO(p+1,q+1)$ to a parabolic subgroup $P \subset \SO(p+1,q+1)$ stabilizing the line generated by $X^I$ --- i.e. $p^I{}_J$ is in $P$ if and only if $p^I{}_J X^J$ is proportional to $X^I$. 

This should be compared with the so-called MacDowell-Mansouri formulation \cite{macdowell_unified_1977} of 4-dimensional General Relativity, in which a preferred direction in $\R^{p+1,q}$ is singled out, thus reducing the local symmetry group of the theory from $\SO(p+1,q)$ to $\SO(p,q)$ --- this aspect is specially clear in the treatment of \cite{stelle_sitter_1979}. We note that this is ultimately related to the theory of Cartan connections, see \cite{wise_macdowellmansouri_2010}, which in fact describes the underlying geometry of tractors \cite{thomas_conformal_1926}.

The elements in the parabolic group $P$ can be parametrized as
\begin{align}\label{Introduction Tractor Functional P-Invariance}
	p^I{}_J=\begin{pmatrix}1 & 0 & 0 \\ r^i & \delta^i{}_k & 0 \\ -\frac{1}{2}|r|^2 & -r_k & 1 \end{pmatrix}\begin{pmatrix}\lambda & 0 & 0 \\ 0 & m^k{}_j & 0 \\ 0 & 0 & \lambda^{-1} \end{pmatrix},
	& &
	p^I{}_J X^J=\lambda^{-1}X^I,
\end{align}
where $\lambda\in\R_+$, $m^i{}_j\in\SO(p,q)$ and $r^i\in\R^n$. Note that these correspond to the subgroups of Weyl rescalings, Lorentz transformations and special conformal transformations of $\R^{p,q}$.

The dynamical fields in our tractor formulation are given by a \emph{tractor connection} $D$, here simply thought of as a $\mathfrak{so}\left(p+1,q+1\right)$-valued connection,
\begin{equation*}
D = d+ A^I{}_J=\left(\begin{matrix}\nabla^\tau & -e_j & 0 \\ -\xi^i & \nabla^\omega & e^i \\ 0 & \xi_j & \nabla^{-\tau} \end{matrix}\right),
\end{equation*}
together with an \emph{infinity tractor} $I^I$, taken here for convenience to be a generic tractor with constant norm $|I|^2 = -\tilde{\Lambda}$
\begin{align*}
	I^I =  \begin{pmatrix} \sigma \\ \mu^i\\ -\frac{|\mu|^2+\tilde{\Lambda}}{2\sigma}\end{pmatrix}.
\end{align*}
Note, the component gauge fields $e^i$, $\tau$, $\omega^i{}_j$ and $\xi_j$ can be interpreted  as associated to the generators of translations $P^i$, dilations $D$, Lorentz transformation $J^i{}_j$ and special conformal transformations $K_j$ (see \cite{kaku_gauge_1977,freedman_supergravity_2012} for more on gauge theory for the (super)-conformal group).

With all these definitions at hand, we can immediately write the tractor action as
\begin{equation}\label{Introduction Tractor Functional}
S\left[ A, I \right] = \int_M \;   \sigma^{1-n}\; X^I \; I^J\bigg(  \tfrac{1}{2}F^{KL} + \sigma^{-1}\;DI^K\wedge E^L + \tfrac{1}{n} \sigma^{-2} \; \tilde{\Lambda}\;E^K \wedge E^L \bigg)\wedge(\star E^{n-2})_{IJKL},
\end{equation}
where we also introduce $F^I{}_J:=dA^I{}_J+A^I{}_K\wedge A^K{}_J$, $E^I:=DX^I$ and
\begin{equation*}
\left(\star E^{n-2} \right)_{IJKL} := \frac{1}{(n-2)!} \;\epsilon_{IJKL K_1 ... K_{n-2}}\; E^{K_1}\wedge \cdots\wedge E^{K_{n-2}}.
\end{equation*}

This functional enjoys an extended version of Weyl invariance, now in a manifest form, given by transformations of the form
\begin{align}\label{Introduction P-Gauge Transformations}
I^I \mapsto p^I{}_J I^J,
& &
A^{I}{}_{J}\mapsto p^I{}_K \; A^K{}_L \; (p^{-1})^L{}_J - dp^I{}_K \; (p^{-1})^K{}_J,
\end{align}
where $p^I{}_J$ is given as in \eqref{Introduction Tractor Functional P-Invariance}. This follows almost immediate from the from of the action, given by contractions of tractor fields, except for fact that the background position tractor $X^I$ is not invariant under \eqref{Introduction Tractor Functional P-Invariance}. In particular, this implies that $E^I := DX^I$ cannot transform covariantly. Fortunately, however, the action turns out to also be invariant under shifts $E^I \mapsto E^I + z X^I$, and therefore it is completely insensitive to this problem.

The critical points of this functional have a simple interpretation in terms of tractor geometry (see more on this below). In this subsection however we would like to concentrate on its equivalence with \eqref{Introduction Cartan-Palatini-Weyl}. Opening up the tractor fields in terms of their components in \eqref{Introduction Tractor Functional} one obtains the following functional
\begin{equation}\label{Introduction Tractor Functional Lorentz Invariant Form}
S\left[\omega,e,\sigma,\mu\right] = \int_M \sigma^{2-n} \bigg( \tfrac{1}{2} F_{\omega}{}^{ij} + d_{\omega}\left( \sigma^{-1}\mu^i e^j \right) 
-\sigma^{-2}\tfrac{n-2}{2n}\left(|\mu|^2 + \tilde{\Lambda}\right) e^i\wedge e^j \bigg)\wedge ( \star e^{n-2})_{ij}.
\end{equation}

This looks already very close to the Cartan-Palatini-Weyl functional \eqref{Introduction Cartan-Palatini-Weyl}. In fact, up to the condition $\sigma^{-1}\mu^i = \Upsilon^i=(\sigma^{-1}d\sigma)^i$, and also $ \tilde{\Lambda}  = \tfrac{2\Lambda}{(n-1)(n-2)}$), the two functionals are identical. In comparison to \eqref{Introduction Cartan-Palatini-Weyl}, the functional \eqref{Introduction Tractor Functional Lorentz Invariant Form} contains a single additional field $\mu^i$ --- the components $\tau$ and $\xi_j$ of the tractor connection do not contribute to the functional --- as well as an additional gauge invariance under special conformal transformations. This allows us to achieve $\sigma^{-1}\mu^i = \Upsilon^i$ as a gauge fixing condition, rendering the two functionals gauge equivalent.

More explicitly, note that the tractor functional \eqref{Introduction Tractor Functional Lorentz Invariant Form} is fully invariant under $P$-gauge transformations, which in components takes the form
\begin{align*}
\omega^i{}_j \mapsto \omega^i{}_j+\lambda(e^i r_j - r^ie_j),
& &
e^i \mapsto \lambda e^j,
& &
\sigma \mapsto \lambda\sigma,
& &
\mu^i\mapsto \mu^i+\lambda \sigma r^i.
\end{align*}
Here we restrict to $m^i{}_j=\delta^i{}_j$ for convenience. It thus follows that we can always achieve $\sigma^{-1}\mu^i = \Upsilon^i$ by an appropriate choice of $r^i$. Note, finally, that we also recover the gauge symmetries \eqref{Introduction Cartan-Palatini-Weyl Invariance} as the subgroup of \eqref{Introduction P-Gauge Transformations} which preserves the condition $\sigma^{-1}\mu^i =\Upsilon^i$ --- obtained by setting $r^i = (\lambda^{-2} d\lambda)^i$ in \eqref{Introduction Tractor Functional P-Invariance}.

More invariantly, also note that imposing the gauge fixing condition
\begin{align*}
I^I=\begin{pmatrix}
1 \cr 0 \cr -\tfrac{\tilde \Lambda}{2}
\end{pmatrix},
\end{align*}
immediately reduces the functional \eqref{Introduction Tractor Functional} to the original Cartan-Palatini functional \eqref{Introduction Cartan-Palatini} --- this amounts to setting at the same time $\sigma^{-1}\mu^i =\Upsilon^i$ and $\sigma=1$. Moreover the subgroup of \eqref{Introduction P-Gauge Transformations} preserving this condition is exactly $\SO(p,q)$, the local symmetry group of \eqref{Introduction Cartan-Palatini}.

\begin{DetailedCheckForTransformationRules}	
	\XDnote{
The full $P$-gauge transformation can be written in components as
\begin{gather}
\begin{align*}
\omega^i{}_j\mapsto m^i{}_k\omega^k{}_l(m^{-1})^l{}_j-dm^i{}_k(m^{-1})^k{}_j + \lambda\left(m^i{}_ke^k r_j - r^ie_k(m^{-1})^k{}_j \right),
\end{align*}
\\
\begin{align*}
e^i \mapsto \lambda m^i{}_je^j,
& &
\sigma \mapsto \lambda \sigma,
& &
\mu^i\mapsto m^i{}_j\mu^i+\lambda \sigma r^i ,
\end{align*}
\end{gather}
}

	\XDnote{
		We here show that the gauge fixing condition $\sigma^{-1}\mu^i =\Upsilon^i$ is preserved if and only if $r_i = (\lambda^{-1} d\lambda)_i$.
		
		On the one hand,
		\begin{align*}
		\hat{\Upsilon}_i :&= \hat{e}^{-1}_i \left(\hat{\sigma}^{-1} d\hat{\sigma}\right) = \lambda^{-1} \; e^{-1}_i \left( \sigma^{-1} d \sigma + \lambda^{-1} d\lambda \right) \\
		&= \lambda^{-1} \left( \Upsilon_i + (\lambda^{-1} d\lambda)_i \right).
		\end{align*}
		On the other hand,
		\begin{align*}
		\hat{\sigma}^{-1} \hat{\mu}_i :& =  \lambda^{-1} \; \sigma^{-1} \left( \mu^i + \lambda \sigma r^i \right) \\
		&= \lambda^{-1} \left(  \sigma^{-1} \mu^i + \lambda r^i \right).
		\end{align*}
		Therefore $\Upsilon_i = \sigma^{-1} \mu_i$ implies $\hat{\Upsilon}_i = \hat{\sigma}^{-1} \hat{\mu}_i$ if and only if $r_i = (\lambda^{-2} d\lambda)_i$.
	}
\end{DetailedCheckForTransformationRules} 

\subsection{First order tractor functionals for General Relativity}
\mbox{}

We now offer a more geometric description of our approach to tractor calculus and of the main results achieved in this paper.

Tractor calculus can be understood as a generalization of the usual Ricci calculus on pseudo-Riemannian manifolds to the context of conformal geometry, see e.g. \cite{thomas_conformal_1926,curry_introduction_2018}. While a pseudo-Riemannian metric defines a unique metric preserving torsion-free connection on the tangent bundle, its Levi-Civita connection, a conformal metric defines a unique \emph{normal tractor connection} on a canonical conformally invariant \emph{standard tractor bundle} \cite{thomas_conformal_1926}. Here, by tractor bundle we mean a rank $(n+2)$ vector bundle $\cT^S\to M$, endowed fiber metric $h$ of signature $(p+1,q+1)$ together with the so-called \emph{position tractor} $X\in\Gamma(LM\otimes\cT^S)$, a preferred nowhere vanishing null section twisted by the (standard) scale bundle $LM=|\bigwedge^nT^*M|^{-1/n}$ of $M$. A normal tractor connection is then simply a metric connection $D^S:\Gamma(\cT^S)\to\Gamma(T^*M\otimes\cT^S)$ on $\cT^S$ satisfying certain curvature conditions which generalize the usual torsion-free condition.

The existence of a canonical tractor bundle with a unique normal connection is certainly important in the study of conformal geometry, in particular for the construction of conformal invariants \cite{gover_conformally_2003}. But it also provides remarkably efficient tools for the study of conformally compact Einstein manifolds \cite{gover_almost_2010,curry_introduction_2018}. In fact, both the Einstein equation and the conformally compact boundary condition find their simplest form in the language of tractors: Given a conformal manifold, with its standard tractor bundle, representative metrics $\gt$ in the conformal class can be put in one-to-one correspondence with a special class of tractor sections $I^{\gt}\in\Gamma(\cT^S)$. The metric $\gt$ is then conformally compact if and only if at $\partial M$ the position tractor $X$ is both orthogonal and transversal to $I^{\gt}$; the metric $\gt$ is Einstein if and only if $I^{\gt}$ is covariantly constant with respect to the normal tractor connection,
\begin{align*}
D^S I^{\gt}=0.
\end{align*}
This reformulation of Einstein equations is also closely related to Friedrich's conformal field equations \cite{frauendiener_local_1999}, which were important in the derivation of global existence results in General Relativity \cite{friedrich_cauchy_1983,friedrich_peeling_2018,friedrich_geometric_2015} and also found applications in numerical relativity \cite{frauendiener_conformal_2004}.

In the present article we wish to describe how a \emph{first order approach} to tractor geometry can be used in the construction of action functionals for General Relativity, providing manifestly Weyl invariant variational principles for Einstein metrics. These can be immediately applied to conformally compact manifolds and naturally extend the Cartan-Palatini-Weyl variational principle \eqref{Introduction Cartan-Palatini-Weyl}.

Our chosen approach to tractor Einstein geometry follows similar lines to the Cartan-Palatini formulation of General Relativity. Accordingly, we start by fixing an abstract $(n+2)$-rank vector bundle $\cT\to M$, endowed with a fiber metric $h_{IJ}$ of signature $(p+1,q+1)$ and a preferred position tractor $\bm X^I\in\Gamma(\cT\otimes\cL)$, and consider choices of 1) a generalized conformal frame $\bm E^I\in\Gamma(T^*M\otimes\cL\otimes\cT)$, 2) a general tractor connection $D:\Gamma(\cT)\to\Gamma(T^*M\otimes\cT)$ and 3) a non-degenerate tractor section $I^I\in\Gamma(\cT)$. Note that, here, the position tractor is twisted by an abstract line bundle $\cL\to M$ associated to $\cT$.

Together, the generalized conformal frame $\bm E^I$ and the tractor section $I^I$ define a pseudo-Riemannian metric on $M$ as the pull-back of the tractor metric $h_{IJ}$, appropriately rescaled by powers of $\bm \sigma=\bm X_I I^I$,
\begin{align*}
g=\bm\sigma^{-2}\bm E_I\bm E^I.
\end{align*}
In particular, this provides an identification between the abstract tractor bundle $\cT$ and the standard tractor bundle $\cT^S$ for the conformal class defined by $g$.

In this language, the conformally compact boundary conditions for $g$ are simply given by
\begin{align*}
\bm X_I I^I|_{\partial M}=0,
& &
|I|^2|_{\partial M}\neq 0,
\end{align*}
while the Einstein equation for the composite metric $g$ can then be written as a set of first order differential equations on $D$, $\bm E^I$ and $I^I$
\begin{align}\label{Introduction First Order Tractor Einstein Equation}
F^I{}_J\bm X^J=0,
& &
\Ric(F)=0,
& &
\bm\sigma^{-1}\bm E^I=D(\bm\sigma^{-1}\bm X^I),
& &
DI^I=0.
\end{align}
Here, the first three equations simply identify $D$ as the pull-back of the standard normal tractor connection on $\cT^S$, while the final equation is the tractor version of the Einstein equation described above.

The main result of our present paper is the description of a variational principle for obtaining such first order tractor Einstein equations. Up to boundary terms, this can be written simply as
\begin{align}\label{Introduction Tractor Functional Full}
S[D,I,\bm E]=\int_M\bm L^{IJKL}\wedge(\star \bm E^{n-2})_{IJKL},
\end{align}
with
\begin{align*}
\bm L^{IJKL}=\bm\sigma^{1-n}\bm X^I\Big[I^J\Big(\tfrac{1}{2}F^{KL}+\alpha\bm\sigma^{-2}\bm E^K\wedge\bm E^L\Big) + \bm Y^JD\Big(I^K-|I|^2\bm\sigma^{-1}\bm X^K\Big)\wedge\bm E^L\Big].
\end{align*}
Here, $\alpha$ is a polynomial function of the tractor norm $|I|^2$, whose specific form is not too important at the moment, and $\bm Y^I$ is the unique tractor satisfying
\begin{align*}
\bm X_I\bm Y^I=1,
& &
\bm E_I\bm Y^I=0,
& &
\bm Y_I\bm Y^I=0.
\end{align*}

Our first result describes the critical points of the tractor action \eqref{Introduction Tractor Functional Full} in relation to Einstein geometry.
\begin{theorem}
A triple $(D,\bm E^I,I^I)$ is a critical point of \eqref{Introduction Tractor Functional Full} if and only if it satisfies the first order tractor Einstein equations \eqref{Introduction First Order Tractor Einstein Equation}.
\end{theorem}

Our second result describes the relation between the Cartan-Palatini-Weyl functional \eqref{Introduction Cartan-Palatini-Weyl} and the tractor functional \eqref{Introduction Tractor Functional Full}, as described in the previous subsection.
\begin{theorem}
The action functionals \eqref{Introduction Cartan-Palatini-Weyl}, \eqref{Introduction Tractor Functional} and \eqref{Introduction Tractor Functional Full} are gauge equivalent: they define the same moduli space of critical points and, moreover, their values coincide when evaluated on the moduli space.
\end{theorem}

It should be emphasized that in our approach the abstract tractor bundle $\cT$, with its tractor metric $h_{IJ}$ and its position tractor $\bm X^I$, is taken to be fixed. On the other hand, the tractor connection $D$, the generalized conformal frame $\bm E^I$ and the tractor section $I^I$ are varied as part of a variational principle. Consequently, $(\cT,h_{IJ},\bm X^I)$ are not defined in terms of a conformal geometry as is usual in the standard literature on parabolic geometry, but rather it is given \emph{a priori}, as a $P$-vector bundle satisfying favorable topological conditions. Upon imposing the equations of motion, we can identify the tractor bundle $\cT$ with the standard tractor bundle for the conformal metric determined by the (now normal) tractor connection $D$. In other words, the tractor bundle is here seen as a background on which to study conformal geometry, rather than a structure defined by the conformal geometry itself.

In the context of variational principles, this approach presents the important advantage of allowing to consider variations of conformal frames (conformal metrics) without the need of varying the associated tractor bundle. The price to pay for this is the introduction of additional gauge freedom: the isomorphism between the abstract and the standard tractor bundle is determined only up to automorphisms of $\cT$. Consequently, while the normal tractor connection is unique as a connection on the standard tractor bundle, it defines a family of gauge-equivalent connections on the abstract tractor bundle. We will come back on the relationship between the abstract and standard tractor bundle in subsection \ref{Subsec: Scale tractor and the Einstein equation}.

\section{A review of Tractor geometry}

In this section we provide a review of Tractor geometry in a form which is adapted to  variational calculus as needed in the next section. This presentation will in fact closely follow the Cartan-Palatini formulation of General Relativity. We hope that this will benefit physicists willing to learn the subject.

\subsection{Conformal frame bundles}\mbox{}
\subsubsection{Structure group and representations}
Let $\R^{p,q}$ denote the vector space $\R^n = \R^{p+q}$ with a non-degenerate symmetric bilinear form of signature $(p,q)$. Here we use the abstract index notation and write vectors in $\R^{p,q}$ as $V^i$, its bilinear form as $\eta_{ij}$, and their contraction as $V_i=\eta_{ij}V^j\in(\R^{p,q})^*$.

Let $W = \CO(p,q)\subset \GL^+(n)$ denote the \emph{pseudo-Riemannian linear conformal group} which preserves the bilinear form $\eta_{ij}$ up to positive scaling,
\begin{equation*}
W=\Big\{A^i{}_j=\lambda m^i{}_j\in\GL^+(n) \mid \lambda\in\R_+,\; m^i{}_j\in\SO(p,q)\Big\}\subset\GL^+(n),
\end{equation*}
and consider $\cW\to M$ a principal $W$-bundle over an $n$-dimensional manifold $M$.

For every $k\in\R$, denote by $\cL^k\to M$ the associated line bundle defined by the representation
\begin{align*}
W\times\R\to\R,
& &
(\lambda m^i{}_j,\bm s)\mapsto\lambda^k \bm s,
\end{align*}
and denote by $\cE\to M$ the associated vector bundle defined by the representation
\begin{align}\label{eq:weight-zero_tangent_bundle}
	W\times\R^n\to\R^n,
	& &
	(\lambda m^i{}_j,V^i)\mapsto m^i{}_jV^j.
\end{align}

We assume here that the line bundles $\cL^k$ are all trivial and that the vector bundles $\cL^k\otimes\cE$ are isomorphic to $TM$. We will then call $\cL^k$ the \emph{(internal) weight $k$ scale bundle} over $M$ and $\cL^k\otimes\cE$ the \emph{(internal) weight $k$ tangent bundle}.

In order to visually keep track of the \emph{weights}, we will use here an extended version of the abstract index notation where sections with non-trivial weights are represented by boldface letters. For example, we will write
\begin{align*}
\bm\sigma\in\Gamma(\cL),
& &
\mu^i\in\Gamma(\cE),
& &
\bm\rho\in\Gamma(\cL^{-1}),
& &
\bm g\in\Gamma(S^2 T^*M\otimes\cL^2).
\end{align*}

\subsubsection{Conformal frames and Weyl connections}\label{Subsubsec: Conformal frames}
A \emph{conformal frame} on $M$ is defined to be a soldering form on $\cL\otimes\cE$, that is, a bundle isomorphism $\bm e^i:TM\to\cL\otimes\cE$ between the tangent bundle of $M$ and the internal weight 1 tangent bundle.

Note that although $\cE$ is associated to the principal bundle $\cW$, its structure group is reduced to $\SO(p,q)$, see equation \eqref{eq:weight-zero_tangent_bundle}. In particular, $\cE$ is equipped with a well-defined non-degenerate fiber metric of signature $(p,q)$, which with a slight abuse of notation we will also denote by $\eta_{ij}\in\Gamma(S^2\cE^*)$.

Given a conformal frame $\bm e^i$, we can define an $\cL^2$-valued metric on $TM$ via the pull-back
\begin{equation*}
\bm{g}=\eta_{ij}\bm e^i\bm e^j\in\Gamma(S^2T^*M\otimes\cL^2).
\end{equation*}
We will call this $\cL^2$-valued metric $\bm g=\bm{g}[\bm e]$ the \emph{conformal metric} associated to $\bm e^i$. This is justified since, once $\bm g$ is given, there exists a one-to-one correspondence between pseudo-Riemannian metrics $g$ proportional to $\bm g$ and non-degenerate (positive) scales $\bm\sigma\in\Gamma(\cL)$ via
\begin{equation*}
g=\bm\sigma^{-2}\bm g\in\Gamma(S^2T^*M).
\end{equation*}
In other words, $\bm g$ defines a conformal class of metrics on $M$ in the usual sense.

Note that, since $\bm g$ is non-degenerate, it also defines a nowhere vanishing $\cL^n$-valued volume form $\bm{\nu}\in \Gamma(\bigwedge^nT^*M\otimes\cL^n)$. In particular, the bundle $\bigwedge^nT^*M \otimes \cL^{n}$ must be a trivial line bundle. Moreover, since $M$ is orientable and since $\cL$ is assumed trivial, we must also have $\bm{\nu} = \bm{s}^{-n}$ for some non-degenerate section $\bm{s} \in \Gamma(LM\otimes\cL^*)$. The conformal metric $\bm g$ thus also defines an isomorphism between the internal scale bundle $\cL$ to the \emph{standard scale bundle} $LM:=|\bigwedge^nT^*M|^{-1/n}$.

A principal $W$-connection $(\tau,\omega^i{}_j)\in\Omega^1(\cW,\so(p,q)\oplus\R)$ on $\cW$ induces a \emph{Weyl connection} $\nabla^\tau$ on $\cL$ and a \emph{spin connection} $\nabla^\omega$ on $\cE$, preserving the canonical fiber metric $\eta_{ij}$. The Weyl connection $\nabla^\tau$ on $\cL$ extends uniquely to connections on $\cL^k$, for every $k\in\R$, and, together with the spin connection $\nabla^\omega$, to all the tensor products of $\cL^k$ and $\cE$. Given a conformal frame $\bm e^i$, we say that $(\tau,\omega^i{}_j)$ is \emph{torsion-free} with respect to $\bm e^i$ if
\begin{align*}
d^{\tau,\omega}\bm e^i=d\bm e^i+\tau\wedge\bm e^i+\omega^i{}_j\wedge\bm e^j=0.
\end{align*}
This condition can be used to determine the spin connection $\nabla^\omega$ as a function of the conformal frame $\bm e^i$ and of a choice of Weyl connection $\nabla^\tau$. In particular, from the point of view of the principal bundle $\cW$, there is no canonical torsion-free connection associated to a given conformal frame. This should be compared with the uniqueness of the Levi-Civita connection of a pseudo-Riemannian metric.

\subsection{Tractor bundles}\label{Subsec: Tractor bundles}
\subsubsection{Structure group and representations}\label{Subsubsec: Tractor Structure group and representations}

Now consider $G=\SO(p+1,q+1)$ the full \emph{pseudo-Riemannian conformal group} and let $P\subset G$ be the parabolic subgroup stabilizing an oriented null line $N\subset\R^{p+1,q+1}$. Since the subspace $N\subset \R^{p+1,q+1}$ is null, it is contained in its own orthogonal complement $N^\perp\subset \R^{p+1,q+1}$. This subspace is also stabilized by the action of $P$, and thus we obtain a canonical $P$-invariant filtration
\begin{equation*}
N\subset N^\perp\subset \R^{p+1,q+1},
\end{equation*}
and also canonical $P$-invariant quotients $L=\R^{p+1,q+1}/N^\perp$ and $E=N^\perp/N$. The vector space $E$ is then endowed with a $P$-invariant bilinear form of signature $(p,q)$ and the line $L$ is endowed with a $P$-invariant pairing with the subspace $N$, both induced by the bilinear form of $\R^{p+1,q+1}$. In particular, the induced actions of $P$ on $E$, $L$ and $N$ reduce to actions of $W=\R_+\times\SO(p,q)$, respectively isomorphic to the actions on $\R^n$ and $\R$ defined in the previous subsection, and the action on $\R^*$ obtained via duality.

To be more explicit, making an appropriate choice of basis, we can write 
\begin{align*}
V^I=\begin{pmatrix} \bm V^+ \cr V^i \cr \bm V^- \end{pmatrix}\in \R^{p+1,q+1},
&&
h_{IJ}=\begin{pmatrix}0 & 0 & 1 \cr 0 & \eta_{ij} & 0 \cr 1 & 0 & 0 \end{pmatrix},
& &
|V|^2=V_IV^I=2\bm V^+\bm V^-+V_iV^i
\end{align*}
where $\eta_{ij}$ denotes a signature $(p,q)$ bilinear form on $\R^{p,q}$.

Further, we can choose the preferred null line $N\subset \R^{p+1,q+1}$ to be generated by the vector
\begin{align*}
X^I=\begin{pmatrix} 0 \cr 0 \cr 1 \end{pmatrix}\in\R^{p+1,q+1},
& &
N=\Big\{\begin{pmatrix} 0 \cr 0 \cr \bm V^- \end{pmatrix}\mid \bm V^-\in\R\Big\},
& &
N^\perp=\Big\{\begin{pmatrix} 0 \cr V^i \cr \bm V^- \end{pmatrix}\mid \bm V^-\in\R,\, V^i\in \R^{p,q}\Big\},
\end{align*}
and we find the following parametrization of the quotient spaces $L$ and $E$
\begin{align*}
L=\Big\{\begin{pmatrix} \bm V^+ \cr * \cr * \end{pmatrix}\mid \bm V^+\in\R\Big\},
& &
E=\Big\{\begin{pmatrix} 0 \cr V^i \cr * \end{pmatrix}\mid V^i\in\R^{p,q}\Big\}.
\end{align*}
The parabolic subgroup $P\subset G$ can then also be decomposed as $P = (\R^n)^*\ltimes W$, with elements parametrized uniquely as 
\begin{equation*}
p^I{}_J=\begin{pmatrix}1 & 0 & 0 \cr r^i & \delta^i{}_k & 0 \cr -\frac{1}{2}|r|^2 & -r_k & 1 \end{pmatrix}\begin{pmatrix}\lambda & 0 & 0 \cr 0 & m^k{}_j & 0 \cr 0 & 0 & \lambda^{-1} \end{pmatrix}\in P,
\end{equation*}
for $\lambda\in\R_+$, $m^i{}_j\in\SO(p,q)$ and $r_i\in(\R^n)^*$.

Now, given $\cP\to M$ a principal $P$-bundle over $M$, we define the \emph{(internal) tractor bundle} $\cT\to M$ as the rank $n+2$ vector bundle associated to the fundamental representation of $P$ on $\R^{p+1,q+1}$
\begin{align*}
P\times\R^{p+1,q+1}\to\R^{p+1,q+1},
& &
(p^I{}_J,V^I)\mapsto p^I{}_JV^J.
\end{align*}
This is endowed with a canonical metric of signature $(p+1,q+1)$, the \emph{tractor metric}, denoted here by $h_{IJ}\in\Gamma(S^2\cT^*)$, and with a canonical filtration $\cN\subset\cN^\perp\subset\cT$, where the subbundles $\cN$ and $\cN^\perp$ are defined with respect to the induced representations of $P$ on $N$ and $N^\perp$, respectively.

We also define quotient bundles $\cL=\cT/\cN$ and $\cE=\cN^\perp/\cN$, respectively associated to the induced representations of $P$ on $L$ and $E$
\begin{align*}
P\times L\to L,
& &
(p^I{}_J,\bm V^+)\mapsto \lambda \bm V^+,
\end{align*}
and
\begin{align*}
	P\times E\to E,
	& &
	(p^I{}_J,V^i)\mapsto m^i{}_j V^j.
\end{align*}

Note that the canonical metric on $E$ and the canonical pairing between $L$ and $N$ give rise, respectively, to a fiber metric of signature $(p,q)$, again denoted by $\eta_{ij}\in\Gamma(S^2\cE^*)$, and to a nowhere vanishing section $\bm X^I\in\Gamma(\cN\otimes\cL)$. Here, $\bm X^I$ can also be interpreted as a canonical isomorphism $\bm X^I:\cL^{-1}\to\cN$ between the line bundle $\cL^{-1}=\cL^*$ and the null subbundle $\cN\subset\cT$. The non-degenerate section $\bm X^I$ is often referred to as the \emph{position tractor}.

Note also that, although the vector bundles $\cL$ and $\cE$ are associated to the principal $P$-bundle $\cP$, they effectively they correspond to representations of the quotient group $W = P/(\R^n)^*$. While the projection $P\to W$ is canonically defined --- so we can alway define a principal $W$-bundle $\cW=\cP/(\R^n)^*$ --- reductions of the structure group from $P$ to $W$ are non-unique and correspond to the introduction of additional structures on $\cP$.

\subsubsection{Tractor frames and tractor decompositions}\label{Subsubsec: Tractor frames and tractor decompositions}

By construction, the internal tractor bundle $\cT$ is isomorphic to a direct sum $\cL\oplus\cE\oplus\cL^{-1}$, although in a non-canonical way. Different choices of such isomorphisms will be called \emph{tractor decompositions} of $\cT$. They are the structures which correspond to reductions of the structure group of $\cP$ to $W$.

Consider a bundle map $\bm Y^I:\cL\to\cT$ embedding the line bundle $\cL$ into $\cT$ as a null subbundle \emph{complementary} to $\cN$. More precisely, let $\bm Y^I\in\Gamma(\cT\otimes\cL^{-1})$ be a nowhere vanishing section satisfying
\begin{align*}
\bm X_I\bm Y^I=1,
& &
\bm Y_I\bm Y^I=0.
\end{align*}
Such an embedding of $\cL$ into $\cT$ uniquely determines a corresponding embedding of $\cE$ into $\cT$ as an orthogonal subbundle to both $\cN$ and $\cL$, and which projects by the identity onto $\cE=\cN^\perp/\cN$. In other words, given the section $\bm Y^I$, there exists a unique section $Z^{I}{}_j\in\Gamma(\cT\otimes\cE^*)$ satisfying
\begin{align*}
\bm X_I Z^{I}{}_j=0,
& &
\bm Y_I Z^{I}{}_j=0,
& &
Z^i{}_j=\delta^i{}_j,
\end{align*}
where $Z^i{}_j:\cE\to\cE$ is the composition of $Z^{I}{}_j:\cE\to\cN^\perp$ with the canonical projection $\cN^\perp\to\cE$, and $\delta^i{}_j:\cE\to\cE$ denotes the identity map on $\cE$.

The triple $(\bm X^I,\bm Y^I, Z^{I}{}_j)$ then defines a tractor decomposition $\cT\to\cL\oplus\cE\oplus \cL^{-1}$ via
\begin{align*}
I^I\mapsto\begin{pmatrix} \bm\sigma \cr \mu^i \cr \bm \rho \end{pmatrix},
& & 
\bm \sigma:=I^I\bm X_I\in\Gamma(\cL),
\quad
\mu^j :=I^I Z_I{}^j\in\Gamma(\cE),
\quad
\bm \rho :=I^I\bm Y_I\in\Gamma(\cL^{-1}),
\end{align*}
for every section $I^I\in\Gamma(\cT)$. The inverse isomorphism can be explicitly written as
\begin{equation*}
\begin{pmatrix} \bm\sigma \cr \mu^i \cr \bm\rho \end{pmatrix}\mapsto I^I:=\bm X^I \bm\sigma + Z^{I}{}_j\mu^j +\bm Y^I \bm\rho.
\end{equation*}
We will call the section $\bm Y^I$, and with some abuse of notation also the corresponding triple $(\bm X^I,\bm Y^I, Z^{I}{}_j)$ as above, a \emph{tractor frame} on $\cT$. Note that, by construction, these also give a decomposition of the tractor metric via
\begin{equation*}
	h_{IJ}= \bm X_I\bm Y_J+\bm Y_I\bm X_J+Z_I{}_i Z_J{}^i.
\end{equation*}

The space of tractor frames on $\cT$ is an affine space modeled on $\Gamma(\cL^{-1}\otimes\cE^*)$. In fact, given any pair $(\bm Y^I, \hat{\bm Y}^I)$ of tractor frames, we obtain the following relations
\begin{align}\label{tractors_frame_transformations}
\hat{\bm Y}^I=\bm Y^I-Z^{I}{}_j\bm z^{j}-\bm X^I\frac{1}{2}|\bm z|^2,
& &
\hat Z^{I}{}_j= Z^{I}{}_j+\bm X^I\bm z_{j},
& &
\bm z_{j}=\bm Y_I \hat Z^{I}{}_j\in\Gamma(\cL^{-1}\otimes\cE^*).
\end{align}
In particular, we can easily compute \emph{transformations rules} for the components of a tractor section $I^I$ induced by a change of tractor frames
\begin{equation}\label{eq:tractor_transformation_rules}
\begin{pmatrix} \hat{\bm\sigma} \cr \hat\mu^i \cr \hat{\bm\rho} \end{pmatrix}= 
\begin{pmatrix}
	 1 & 0 & 0\\
	 \bm z^{i} & \delta^i{}_j & 0 \\
	 -\frac{1}{2}|\bm z|^2 & -\bm z_{j} & 1
 \end{pmatrix}
 \begin{pmatrix} \bm\sigma \cr \mu^j  \cr \bm\rho \end{pmatrix}.
\end{equation}

Here we see that the first component $\bm\sigma$ of a tractor $I^I$ is completely independent of the choice of tractor frame. In fact, we have a well-defined projection $\bm X_I : \cT\to\cL$. We will call a tractor section $I^I\in\Gamma(\cT)$ \emph{non-degenerate} if it projects to a non-degenerate scale $\bm \sigma\in\Gamma(\cL)$.

Now, given a conformal frame $\bm e^i\in\Gamma(T^*M\otimes\cL\otimes\cE)$, together with a tractor frame $\bm Y^I\in\Gamma(\cL^{-1}\otimes\cT)$, we can define an embedding of the tangent bundle of $M$, twisted by $\cL^{-1}$, into the tractor bundle $\cT$ via
\begin{equation*}
\bm E^I=Z^{I}{}_j\bm e^j\in\Gamma(T^*M\otimes\cL\otimes\cN^\perp).
\end{equation*}
This can then be used to pullback the tractor metric to a conformal metric on $M$ via
\begin{equation*}
\bm g=h_{IJ}\bm E^I\bm E^J=\eta_{ij}\bm e^i\bm e^j.
\end{equation*}
We will thus call $\bm E^I$ a \emph{generalized conformal frame} on $\cT$.

A generalized conformal frame $\bm E^I$ completely encodes the underlying conformal and tractor frames. In fact, in terms of a reference tractor decomposition, this follows easily from the decomposition
\begin{equation*}
\bm E^I = \begin{pmatrix} 0 \cr \bm e^i \cr -\bm z_{j}\bm e^j \end{pmatrix},
\end{equation*}
where $\bm e^i\in\Gamma(T^*M\otimes \cL\otimes\cE)$ is the underlying conformal frame and $\bm z_{j}\in\Gamma(\cL^{-1}\otimes\cE^*)$ parametrizes the underlying tractor frame. More invariantly, $\bm e^i$ can be defined as the projection of $\bm{E}^I$ into $\Gamma(T^*M\otimes\cL\otimes\cE)$, and $\bm Y^I$ as the unique section in $\Gamma(\cL^{-1}\otimes\cT)$ satisfying
\begin{align*}
\bm Y^I\bm X_I=1,
& &
\bm Y^I\bm Y_I=0,
& &
\bm Y^I\bm E_I=0.
\end{align*}

\subsubsection{Tractor connections}\label{Subsubsec: Tractor connections}

We now define a \emph{tractor connection} as a linear connection $D:\Gamma(\cT)\to\Gamma(T^*M\otimes\cT)$ on $\cT$ which preserves the tractor metric $h_{IJ}$.

Fixing a reference tractor decomposition $\cT\to\cL\oplus\cE\oplus\cL^{-1}$, we can write any tractor connection as
\begin{equation*}
D=\begin{pmatrix}\nabla^\tau & -\bm\theta_{j} & 0 \cr -\bm\xi^{i} & \nabla^\omega & \bm\theta^{i} \cr 0 & \bm\xi_{j} & \nabla^{-\tau} \end{pmatrix}
\end{equation*}
where $\nabla^\tau:\Gamma(\cL)\to\Gamma(T^*M\otimes\cL)$ and $\nabla^\omega:\Gamma(\cE)\to\Gamma(T^*M\otimes\cE)$ are Weyl and spin connections induced respectively on $\cL$ and $\cE$, and $\bm\theta^{i}:\cE^*\to TM^*\otimes\cL$ and $\bm\xi_{j}:\cE\to TM^*\otimes\cL^{-1}$ are 1-forms on $M$ with values respectively in $\cL\otimes\cE$ and $\cL^{-1}\otimes\cE^*$.

The transformation rules associated with changing the choice of tractor frame can be easily computed from the relations \eqref{tractors_frame_transformations} and \eqref{eq:tractor_transformation_rules}. We find
\begin{align}\label{Transformation rules for tractor connections1}
	\hat{\bm\theta}_i=\bm\theta_{i},
	& &
	\nabla^{\hat\tau}=\nabla^\tau+\bm\theta_j\bm z^j,
	& &
	\nabla^{\hat\omega}=\nabla^\omega+\bm\theta^{i}\bm z_{j}-\bm z^{i}\bm\theta_{j},
\end{align}
and
\begin{equation*}
	\hat{\bm\xi}_i=\bm\xi_i+\nabla^{\tau,\omega}\bm z_i-\Big(\bm z_i\bm z_j-\frac{1}{2}|\bm z|^2\eta_{ij}\Big)\bm\theta^j=\bm\xi_i-Q(-\bm z_i),
\end{equation*}
where $\bm z_i\in \Gamma(\cL^{-1}\otimes\cE^*)$ and $Q:\Gamma(\cL^{-1}\otimes\cE^*)\to \Gamma(T^*M\otimes\cL^{-1}\otimes\cE^*)$ is a first-order non-linear differential operator defined by
\begin{align}\label{Q operator Def}
Q(\bm r_i)&
=\nabla^{\tau,\omega}\bm r_i+\Big(\bm r_i\bm r_j-\frac{1}{2}|\bm r|^2\eta_{ij}\Big)\bm\theta^j.
\end{align}

\begin{DetailedCheckForTransformationRules}	
	
	\XDnote{These transformation rules can be checked as follows:
		\begin{align*}
		\bm \theta^i :&= Z^i_I D \bm X^I = \hat{Z}^i_I D \bm X^I = \hat{\bm \theta}^i\\ \\	
		\tau :&= -\bm Y_I D \bm X^I = -(\hat{\bm Y} + \bm z_j \hat{Z}^j   )_I D \bm X^I = \hat{\tau} - \bm z_{j} \bm \theta^j\\ \\
		\omega^i{}_j :&= Z^i_I D \bm Z^I_j =  (\hat{Z}^i_I - \bm z^i \bm X_I ) D  (\hat{Z}^I_j - \bm z_j X^I )\\ 
		 &= \hat{\omega}^i{}_j - \bm z^i \bm X_I D \hat{Z}^I_j -  \hat{Z}^i_I D (\bm X^I) \bm z_j  =  \hat{\omega}^i{}_j + \bm z^i \bm \theta_j -  \bm \theta^i \bm z_j\\ \\
		 \bm\xi_i :&= \bm Y_I D Z^I_i =  \bm Y_I D (  \hat{Z}_i - \bm z_i \bm X  )^I = - \nabla^{\tau} \bm z_ i  + \bm Y_I D \hat{Z}_i \\
		 &=  - \nabla^{\tau} \bm z_ i  + (\hat{\bm Y} + \bm z_j \hat{Z}^j- \tfrac{1}{2} |\bm z|^2 X )_ID \hat{Z}_i^I \\
		 &= - \nabla^{\tau} \bm z_ i +  \hat{\bm\xi}_i + \bm z_j \hat{\omega}^j{}_i + \tfrac{1}{2} |\bm z|^2  \bm \theta_i \\
		  &= - \nabla^{\tau,\omega} \bm z_ i +  \hat{\bm\xi}_i + \bm z_j (  -\bm z^j \bm \theta_i +  \bm \theta^j \bm z_i  )  + \tfrac{1}{2} |\bm z|^2  \bm \theta_i \\
		 &= - \nabla^{\tau,\omega} \bm z_ i +  \hat{\bm\xi}_i + \bm z_j z_i \bm \theta^j \bm   - \tfrac{1}{2} |\bm z|^2  \bm \theta_i \\
		 &= Q(-\bm z_i ) +  \hat{\bm\xi}_i
		\end{align*}

	}
\end{DetailedCheckForTransformationRules}

It is important to note from \eqref{Transformation rules for tractor connections1} that the component $\bm\theta^{i}\in\Gamma(T^*M\otimes\cL\otimes\cE)$ of a tractor connection is well-defined independently of the choice tractor frames. We will thus call a tractor connection \emph{non-degenerate} if the corresponding bundle map $\bm\theta^i:\cE^*\to TM^*\otimes\cL$ is a bundle isomorphism. In particular, every non-degenerate tractor connection gives rise to a conformal frame on $M$, and therefore to a corresponding conformal metric.

The \emph{curvature} of a tractor connection is defined as usual as the unique section $F^I{}_J\in\Gamma(\bigwedge^2T^*M\otimes\End(\cT))$ such that
\begin{equation*}
F^I{}_JI^J=d^D DI^I,
\end{equation*}
for every $I^I\in\Gamma(\cT)$. A non-degenerate tractor connection $D$ is called \emph{normal} if, with respect to any tractor frame, 1) $\bm{F}^I{}_+:=F^I{}_J\bm X^I$ is identically zero and 2) $F^i{}_j:=Z_I{}^iF^I{}_JZ^{J}{}_j$ is \emph{Ricci free}.

\begin{proposition}\label{Unicity of Normal Connection}
Let $\bm g$ be a given conformal metric on $M$. Then, up to $P$-gauge transformations,
\begin{itemize}[topsep=0pt,itemsep=10pt,leftmargin=20pt]
\item there exists a unique compatible normal tractor connection (for $n>2$);
\item compatible normal tractor connections form an affine space modeled on $\Gamma(S^2_0T^*M)$ (for $n=2$).
\end{itemize}
\end{proposition}

\begin{proof}
Fixing a reference tractor decomposition, the curvature of a tractor connection $D$ can be written as
\begin{equation*}
F^I{}_J=\begin{pmatrix}F_\tau+\bm\theta_{k}\wedge\bm\xi^k & -d^{\tau,\omega}\bm\theta_{j} & 0 \cr -d^{\tau,\omega}\bm\xi^{i} & F_\omega{}^i{}_j+\bm\theta^{i}\wedge\bm\xi_{j}+\bm\xi^{i}\wedge\bm\theta_{j} & d^{\tau,\omega}\bm\theta^{i} \cr 0 & d^{\tau,\omega}\bm\xi_{j} & -F_\tau-\bm\theta_{k}\wedge\bm\xi^k \end{pmatrix},
\end{equation*}
where $F_\tau$ denotes the curvature of $\nabla^\tau$, $F_\omega{}^i{}_j$ the curvature of $\nabla^\omega$, and $d^{\tau,\omega}$ the exterior covariant derivative on the corresponding associated bundle.

The normality condition is then equivalent to
\begin{align*}
	d^{\tau,\omega}\bm\theta^{i}=0,
	& &
	\bm\xi_{[i}{}_{j]}=-\frac{1}{2}\bm{F}_\tau{}_{ij}=-\frac{1}{n-2}\bm F_\omega{}^k{}_{[i}{}_{|k|j]},
	& &
	\bm\xi^k{}_k=-\frac{1}{2(n-1)}\bm{F}_\omega{}^{kl}{}_{kl},
\end{align*}
\begin{equation}\label{normality_conditions}
	\bm\xi_{(i}{}_{j)}-\frac{1}{n}\bm\xi^k{}_k\eta_{ij}=-\frac{1}{n-2}\Big(\bm{F}_\omega{}^k{}_{(i}{}_{|k|j)}-\frac{1}{n}\bm{F}_\omega{}^{kl}{}_{kl}\eta_{ij}\Big).
\end{equation}

Note that, for $n>2$ these completely determine $\nabla^\omega$ and $\bm\xi_{ij}$ in terms of $\nabla^\tau$ and $\bm\theta^{i}$, while for $n=2$ the symmetric trace-free part of $\bm\xi_{ij}$ is left undetermined.

Let $D$ and $\hat D$ denote a pair of normal tractor connections compatible with a conformal metric $\bm g$, that is, such that
\begin{align*}
\bm g=\eta_{ij}\bm\theta^i\bm\theta^j=\eta_{ij}\hat{\bm\theta}^i\hat{\bm\theta}^j.
\end{align*}
Up to a $P$-gauge transformation we can assume
\begin{align*}
\hat{\bm\theta}^i=\bm\theta^i,
& &
\nabla^{\hat\tau}=\nabla^\tau,
\end{align*}
and the normality condition determines $\nabla^{\hat\omega}=\nabla^\omega$ and $\hat{\bm\xi}_i\hat{\bm\theta}^i-\bm\xi_i\bm\theta^i=q$, with $q\in\Gamma(S^2_0 T^*M)$ identically zero for $n>2$ and arbitrary for $n=2$.
\end{proof}

This result dates all the way back to Cartan \cite{cartan_les_1923} and Thomas \cite{thomas_conformal_1926}, see also \cite{bailey_thomass_1994}. It provides (for $n>2$) a canonical normal tractor connection on $\cT$ associated with a given conformal frame (conformal metric). Such normal Cartan connection in conformal geometry thus plays an equivalent role to the spin (Levi-Civita) connection in Riemannian geometry.

\subsection{Scale tractors and the tractor Einstein equation}\label{Subsec: Scale tractor and the Einstein equation}

\subsubsection{Scale tractors and induced decompositions}
Given a non-degenerate tractor connection $D:\Gamma(\cT) \to \Gamma(T^*M\otimes \cT)$ on a tractor bundle $\cT \to M$ we can define a preferred class of tractors, associated with non-degenerate scales. More precisely, a choice of non-degenerate tractor connection $D$ defines a differential operator $T^I:\Gamma(\cL)\to\Gamma(\cT)$, assigning to each non-degenerate scale ${\bm\sigma}$ a \emph{scale tractor} $T^I({\bm\sigma})$ defined uniquely by the conditions
\begin{align*}
\bm X_IT^I({\bm\sigma})={\bm\sigma},
& &
\bm X_I DT^I({\bm\sigma})=0,
& &
Z_I{}^i\bm D_i T^I({\bm\sigma})=0,
\end{align*}
for some --- and therefore for any --- tractor frame $(\bm X^I,\bm Y^I,Z^I{}_j)$. In a reference tractor decomposition we can write the scale tractor as
\begin{align*}
	T^I({\bm\sigma}) = {\bm\sigma} \begin{pmatrix} 1 \\ \bm\Upsilon^i \\ -\frac{1}{2}|\bm \Upsilon|^2-\frac{1}{n}\big(\bm Q_k(\bm\Upsilon^k)-\bm\xi_k{}^k\big) \end{pmatrix}
\end{align*}
with $\bm\theta^k\bm\Upsilon_k = \bm\sigma^{-1}\nabla^\tau\bm\sigma$ and $Q = \bm\theta^k\bm Q_k$ defined as in \eqref{Q operator Def}.

In turn, we can use the scale tractor $T^I({\bm\sigma})$ to define a preferred tractor frame $(\bm X^I,\bm Y_{\bm\sigma}{}^I,Z_{\bm\sigma}{}^I{}_j)$
\begin{align*}
\bm Y_{\bm\sigma}{}^I:={\bm\sigma}^{-1}T^I({\bm\sigma})-\tfrac{1}{2}|\bm\sigma^{-1}T({\bm\sigma})|^2\bm X^I = \begin{pmatrix} 1 \\ \bm\Upsilon^i \\ -\frac{1}{2}|\bm\Upsilon|^2 \end{pmatrix},
& &
\bm E_{\bm\sigma}{}^I = Z_{\bm\sigma}{}^I{}_j\bm\theta^j := D^{\tau_{\bm\sigma}}\bm X^I = \begin{pmatrix} 0 \\ \bm\theta^i \\ -\bm\Upsilon_k\bm\theta^k \end{pmatrix},
\end{align*}
where $D^{\tau_{\bm\sigma}}\bm X^I:=\bm\sigma D(\bm\sigma^{-1}\bm X^I)$, and we obtain an induced tractor decomposition of $\cT$ for each non-degenerate scale
\begin{align*}
A(\bm\sigma):\cT\to\cL\oplus\cE\oplus\cL^{-1}.
\end{align*}
Note that such a decomposition can be equivalently defined by requiring the that the induced Weyl connection $\nabla^\tau = \nabla^{\tau_{\bm\sigma}}$ on $\cL$ satisfies $\nabla^{\tau_{\bm\sigma}}\bm\sigma = \nabla^\tau\bm\sigma -\bm\sigma\bm\Upsilon_k\bm\theta^k = 0$. Also note that the transformation rules between decompositions corresponding to distinct choices of non-degenerate scales are exactly the same as the transformation rules in the usual definition of the standard tractor bundle from \cite{bailey_thomass_1994}
\begin{equation*}
	A(\hat{\bm\sigma})\circ A(\bm\sigma)^{-1}= 
	\begin{pmatrix}
		 1 & 0 & 0\\
		 \bm z^{i} & \delta^i{}_j & 0 \\
		 -\frac{1}{2}|\bm z|^2 & -\bm z_{j} & 1
	 \end{pmatrix}
\end{equation*}
with $\hat{\bm\sigma} = \lambda \bm\sigma$ and $\bm\theta^k\bm z_k = -\lambda^{-1} d\lambda$.

\begin{DetailedCheckForTransformationRules}
	
\XDnote{These transformation rules can be checked as follows:
	They must be of the form \eqref{eq:tractor_transformation_rules} with $\bm z_i$ determined by the reasoning below.

On the one hand we have	
\begin{align*}
\hat{\bm E}^I :&= \hat{\bm \sigma} D (\hat{\bm\sigma}^{-1} \bm X)^I = \bm \sigma D (\bm\sigma^{-1} \bm X)^I - \lambda^{-1} d\lambda \; \bm X^I \\ &= \bm E^I - \lambda^{-1} d\lambda \; \bm X^I.
\end{align*}
On the the other hand
\begin{align*}
\hat{\bm E}^I :&= \hat{Z}^I_i \bm \theta^i = ( Z^I_i + \bm X^I \bm z_i ) \bm \theta^i \\
&= \bm E^I + (\bm z_i \bm \theta^i) \bm X^I.
\end{align*}
Comparing the two expressions gives $\bm z_i \bm \theta^i = - \lambda^{-1} d\lambda$.
}

\XDnote{ One obtains the transformation rules from Bailey, Eastwood and Gover \cite{bailey_thomass_1994} (p 1197) by the replacement $\lambda \mapsto \lambda^{-1}$. This is coherent with a previous remark and the fact that with $\hat{{\bm\sigma}} = \lambda {\bm\sigma}$ one has
	\begin{equation*}
	\hat{g} := \hat{\bm\sigma}^{-2} \bm{g} =  \lambda^{-2} \bm\sigma^{-2} \bm{g} = \lambda^{-2} g.
	\end{equation*}
	which is the opposite convention than the one taken in \cite{bailey_thomass_1994}.
}	
	 
\end{DetailedCheckForTransformationRules}

\subsubsection{Tractor-Einstein equations}\label{Subsubsec: Tractor-Einstein equation}

So far, we have seen that conformal metrics $\bm g$ on $M$ correspond to normal tractor connections $D$ on the tractor bundle $\cT$ (together with additional data if $n=2$). We have also seen that metric representatives $g$ of $\bm g$ correspond to non-degenerate scales $\bm\sigma$, and that these give rise to scale tractors $I^I=T^I(\bm\sigma)$. This allows for a full description of pseudo-Riemannian geometry in terms of tractor calculus/tractor geometry, which is particularly powerful in the study of Einstein geometry.

\begin{theorem}[see Bailey-Eastwood-Gover \cite{bailey_thomass_1994} and Calderbank \cite{calderbank_mobius_2006,burstall_conformal_2010}, for $n=2$]
Einstein metrics $g$ on $M$ are in one-to-one correspondence with pairs $(D,I^I)$ consisting of 1) a normal tractor connection $D$ and 2) a non-degenerate covariantly constant tractor section $I^I$, considered up to $P$-gauge transformations (and affine translations of $D$ if $n=2$).
\end{theorem}
\begin{proof}	
Let $D$ be a non-degenerate tractor connection and $I^I$ a non-degenerate covariantly constant tractor section. Since $I^I$ is covariantly constant, in particular its norm is a constant, we will write $ \tilde{\Lambda} = -I^2$. In a reference tractor decomposition,
\begin{align*}
	D&=\begin{pmatrix}\nabla^\tau & -\bm\theta_{j} & 0 \cr -\bm\xi^{i} & \nabla^\omega & \bm\theta^{i} \cr 0 & \bm\xi_{j} & \nabla^{-\tau} \end{pmatrix},
	& &
	I^I = \begin{pmatrix} \bm\sigma \\ \mu^i\\ -\tfrac{|\mu|^2+\tilde{\Lambda}}{2\bm\sigma}\end{pmatrix}.
\end{align*}
and
\begin{equation}\label{Tractor Einstein equations 1}
	DI^I = \begin{pmatrix} \nabla^\tau \bm\sigma - \bm\theta_i \mu^i \\ \nabla^{\omega} \mu^i - \bm\xi^i \bm\sigma -\tfrac{ |\mu|^2+\tilde{\Lambda}}{2\bm\sigma} \bm\theta^i \\
\mu^i \bm\xi_i + \tfrac{1}{2} \bm\sigma^{-2} \nabla^{\tau} \bm\sigma ( |\mu|^2+\tilde{\Lambda}) + \tfrac{1}{2} \bm\sigma^{-1} d|\mu|^2
 \end{pmatrix} =0
\end{equation}
introducing the notation $\bm \Upsilon_i\bm\theta^i := \bm\sigma^{-1} \nabla^{\tau} \bm\sigma$, the first two equations are equivalent to
\begin{align}\label{Tractor Einstein equations 2}
	\mu^i = \bm\sigma \bm\Upsilon^i,
	& &
	Q \left ( \bm\Upsilon_i \right) -\bm\xi_i-\tfrac{1}{2} \tilde{\Lambda} \bm\sigma^{-2} \bm\theta_i= 0,
	\end{align}
where $Q:\Gamma(\cL^{-1}\otimes\cE^*)\to \Gamma(T^*M\otimes\cL^{-1}\otimes\cE^*)$ is the first-order non-linear differential operator given by \eqref{Q operator Def}. It only takes a direct calculation to check that when \eqref{Tractor Einstein equations 2} hold the third equation in \eqref{Tractor Einstein equations 1} follows identically. 
 All in all the requirements that $I^I$ is non degenerate and covariantly constant are strictly equivalent to equations \eqref{Tractor Einstein equations 2}. In particular this set of equations does not depend on the choice of tractor decomposition. We will now make use of this fact to pick a particularly adapted tractor frame.

Since, $\bm\sigma=\bm X_II^I$ is a non-degenerate scale, it defines a preferred generalized conformal frame
\begin{equation*}
	\bm E_{\bm\sigma}{}^I=D^{\tau_{\bm\sigma}}\bm X^I:=\bm\sigma D(\bm\sigma^{-1}\bm X^I).
\end{equation*}
Decomposing the tractor bundle with respect to this preferred frame we have $\nabla^{\tau}=\nabla^{\tau_{\bm\sigma}}$. In particular $\bm\Upsilon_i =0$ and equations \eqref{Tractor Einstein equations 2} now take the simple form
\begin{align*}
	\mu^i =0, & & \bm\xi^{i}=-\frac{\tilde{\Lambda}}{2}\bm\sigma^{-2}\bm\theta^{i}.
\end{align*}
Finally, the normality conditions \eqref{normality_conditions} for $D$ in this particular frame become:
\begin{align*}
d^{\tau,\omega}\bm\theta^{i}=0,
& &
\bm{F}_{\tau}=0,
& &
\bm{F}_{\omega}{}^{kl}{}_{kl}=n(n-1)\tilde{\Lambda}\bm\sigma^{-2},
\end{align*}
\begin{equation*}
\bm{F}_\omega{}^k{}_{(i}{}_{|k|j)}-\frac{1}{n}\bm{F}_\omega{}^{kl}{}_{kl}\eta_{ij}=0.
\end{equation*}
All in all, these imply that $\nabla^{\omega}$ and $\nabla^\tau$ are induced by the Levi-Civita connection of the pseudo-Riemannian metric $g=\bm\sigma^{-2}\eta_{ij}\bm\theta^{i}\bm\theta^{j}$, which furthermore satisfies the Einstein equation with scalar curvature $R=n(n-1)\tilde{\Lambda}$.
\end{proof}

\begin{corollary}
In terms of tractor geometry, the scalar curvature of an Einstein metric $g$ is proportional to the norm of the corresponding covariantly constant tractor
\begin{align*}
R=-n(n-1)|I|^2.
\end{align*}
\end{corollary}

\subsubsection{The standard tractor bundle}\label{Subsubsec: The standard tractor bundle}
Let us end this section by discussing how the definition of the tractor bundle presented above and the more usual definition found in \cite{bailey_thomass_1994}.

The \emph{standard tractor bundle} $\cT^S\to M$ is usually defined as a quotient of the second order jet bundle $\cJ^2(LM)$ of standard scales, canonically determined by a choice of conformal metric $\gt\in\Gamma(S^2T^*M\otimes LM^2)$ on $M$ (see e.g. \cite{bailey_thomass_1994}). It is endowed with a canonical tractor metric and a canonical filtration $\cT^S\supset(\cN^S)^\perp\supset \cN^S$, where $\cN^S$ is canonically isomorphic to $LM^{-1}$ and $(\cN^S)^\perp$ (non-canonically) isomorphic to $TM\otimes LM^{-1}$.

The essential feature of the standard definition of the tractor bundle, which distinguishes it from our definition above, is the presence of a canonical differential operator $T^S:\Gamma(LM)\to \Gamma(\cT^S)$ prolonging non-degenerate scales $\sigma\in\Gamma(LM)$ to non-degenerate tractors $T^S(\sigma)\in\Gamma(\cT^S)$, and, accordingly, defining isomorphisms
\begin{equation*}
A^S(\sigma):\cT^S\to LM\oplus (TM\otimes LM^{-1})\oplus LM^{-1}.
\end{equation*}
This is the so called the \emph{Thomas operator} of the standard tractor bundle. As a consequence, the standard tractor bundle comes equipped with a canonical normal tractor connection $D^S$: this is the unique normal tractor connection compatible both with the Thomas operator and with the underlying conformal metric, see \cite{bailey_thomass_1994,curry_introduction_2018} for details.

Let $\bm g$ be a conformal metric and $\cT^S$ the associated standard tractor bundle. Clearly any tractor bundle $\cT$ as defined in subsection \ref{Subsubsec: Tractor Structure group and representations} is isomorphic to $\cT^S$ and this isomorphism is unique up to $P$-gauge transformation. On the other hand, from proposition \ref{Unicity of Normal Connection} normal tractor connections on $\cT$ are also unique up to $P$-gauge transformation. The proposition below asserts that these two freedoms are equivalent.

\begin{proposition}
Let $D$ be a normal tractor connection on $\cT$. It defines an isomorphism $\cT \to \cT^S$ where $\cT^S$ is the standard tractor bundle associated to the conformal metric induced by $D$. Moreover, the isomorphism pulls-back the standard tractor connection $D^S$ on $\cT^S$ to the connection $D$ on $\cT$.
\end{proposition}

\begin{proof}
Given a normal tractor connection $D$ on $\cT$, let $\bm g=\eta_{ij}\bm\theta^i\bm\theta^j\in\Gamma(S^2 T^*M\otimes\cL^2)$ be the associated conformal metric and $\bm s\in\Gamma(LM\otimes\cL^{-1})$ the isomorphism defined in subsection \ref{Subsubsec: Conformal frames}. Define an ``true'' conformal metric $\tilde g=\bm s^2\bm g$, and let $\cT^S$ denote the corresponding standard tractor bundle.

For each non-degenerate scale $\bm\sigma\in\Gamma(\cL)$ we have an ``true'' non-degenerate scale $\sigma=\bm s\bm\sigma\in\Gamma(LM)$ and we can construct tractor decompositions $A(\bm\sigma):\cT\to\cL\oplus\cE\oplus\cL^{-1}$ and $A^S(\sigma):\cT^S\to LM\oplus TM\otimes LM^{-1}\oplus LM^{-1}$. Then, using the isomorphisms $\bm\theta^i: TM \to \cL \otimes \cE$ and $\bm s: LM^{-1}\to\cL^{-1}$ we obtain an isomorphism $B^D:\cT\to\cT^S$. The key point here is that $B^D$ is independent the choice of scale $\bm\sigma$. This follows from the transformation rules $A(\bm\sigma) \mapsto A(\hat{\bm\sigma})$ (see the end of subsection \ref{Subsubsec: Tractor connections}), which are precisely the same as the standard transformation rules $ A^S(\bm\sigma) \mapsto A^S(\hat{\bm\sigma})$ for the standard tractor bundle (see \cite{bailey_thomass_1994}). By construction, $B^D$ pulls-back the standard Thomas operator $T^S:\Gamma(LM)\to\Gamma(\cT^S)$ to the differential operator $T^I:\Gamma(\cL)\to\Gamma(\cT)$, and in particular, the standard tractor connection $D^S$ on $\cT^S$ to the given tractor connection $D$ on $\cT$. 
\end{proof}

As a closing remark, we wish to emphasize that the description of the tractor bundle that was given in the last two subsections is in complete parallel with the so-called Cartan-Palatini formulation of general Relativity. There, an abstract ``internal'' tangent bundle $\cE\to M$ is constructed a priori, as an associated bundle to an abstract ``internal'' principal $\SO(p,q)$-bundle over $M$, and then identified with $TM$ through the choice of an orthogonal frame $e^i\in\Gamma(T^*M\otimes\cE)$. This induces a metric $g = \eta_{ij} e^i e^j$ on $TM$ via pull-back of the metric $\eta_{ij}$ on $\cE$. There are then infinitely many metric-compatible (i.e. satisfying $\nabla \eta_{ij}=0$) connections on $\cE$  but the choice of an orthogonal frame $e^i$ singles out a preferred one: The connection obtained by the pull-back of the Levi-Civita connection $\nabla^{g}$ of $g$ from $TM$ to $\cE$ or, equivalently, as the unique connection such that its tensor product with the Levi-Civita connection satisfies $\nabla e^i = 0$.
It follows that $\SO(p,q)$-gauge transformations act non-trivially on $e^i$ and $\nabla$, but preserve the metric $g$ and its Levi-Civita connection $\nabla^{g}$.

\section{Tractor functionals}

This section contains the main new results of the present paper. We will describe three new action functionals for General Relativity based on Tractor geometry as reviewed in the previous section. Moreover we will describe the equivalence between our new Tractor actions and the Cartan-Palatini-Weyl action functional described in the introduction.

\subsection{Gauge invariant Lagrangians}
\mbox{}

Let $M$ be an oriented manifold and $\cP$ a principal $P$-bundle over $M$, and let $(\cT,h_{IJ},\bm X^I)$ denote the associated tractor bundle, $\cL$ the associated scale bundle, $(\cE,\eta_{ij})$ the associated weight-zero tangent bundle as introduced in subsection \ref{Subsec: Tractor bundles}.

We now wish to describe a construction of gauge invariant functionals on (the space of sections of associated bundles to) $\cP$, defined by integration of Lagrangian densities over $M$. By \emph{Lagrangian density} here, we mean a functional of 1) a non-degenerate tractor connection $D$, 2) a non-degenerate tractor $I^I$, and 3) a generalized conformal frame $\bm E^I$, taking values on $n$-forms on $M$, and which is invariant under $P$-gauge transformations (at least modulo exact terms).  

We will consider Lagrangian densities of the form
\begin{equation*}
\mathscr{L}=\bm L^{I_1\cdots I_{k+2}}\wedge(\star \bm E^{n-k})_{I_1\cdots I_{k+2}}\in\Gamma(\textstyle{\bigwedge^n}T^*M),
\end{equation*}
where $\bm L^{I_1\cdots I_{k+2}}\in\Gamma(\bigwedge^{k}TM^*\otimes\cL^{k-n}\otimes\bigwedge^{k+2}\cT)$ is  a tractor valued weighted $k$-form, depending on the tractor connection $D$, the tractor section $I^I$, and the generalized conformal frame $\bm E^I$. Here we also make use of the \emph{internal star operator} $\star:\Gamma(\bigwedge^{n-k}\cT)\to\Gamma(\bigwedge^{k+2}\cT^*)$ and write
\begin{equation*}
(\star \bm E^{n-k})_{I_1\cdots I_{k+2}}=\frac{1}{(n-k)!}\epsilon_{I_1\cdots I_{n+2}}\bm E^{I_{k+3}}\wedge \cdots \wedge \bm E^{I_{n+2}},
\end{equation*}
with $\epsilon_{I_1\cdots I_{n+2}}$ the \emph{internal volume form} along the fibers of $\cT$.

Note that, being constructed in terms of wedge products and contractions of tractor valued forms, the $n$-forms $\mathscr{L}$ are manifestly invariant under $P$-gauge transformations --- in particular under Weyl rescalings.

To relate them to more familiar quantities, allowing for a better interpretation of their geometric meaning, we can introduce a reference tractor frame and describe $\mathscr{L}$ in terms of the corresponding components of $D$, $I^I$ and $\bm E^I$. This has the effect of ``breaking'' the manifest $P$-invariance into more usual (still manifest) $W$-invariance and more complicated (non-linear) affine-invariance.

Effectively, this can be achieved from decompositions of the $(n-k)$-forms $(\star\bm E^{n-k})_{I_1\cdots I_{k+2}}$, induced by each choice of tractor frame. We find
\begin{align*}
(\star \bm E^{n-k})_{I_1\cdots I_{k+2}}
=\frac{(k+2)!}{k!}\bm Y_{[I_1}\bm X_{I_2}Z_{I_3}{}^{i_1}\cdots Z_{I_{k+2}]}{}^{i_k}(\star\bm e^{n-k})_{i_1\cdots i_k}
\end{align*}
where 
\begin{equation*}
(\star\bm e^{n-k})_{i_{1}\cdots i_{k}}=\frac{1}{(n-k)!}\epsilon_{i_1\cdots i_{n}}\bm e^{i_{k+1}}\wedge \cdots \wedge \bm e^{i_{n}},
\end{equation*}
is defined with respect to the internal Hodge star operator along the fibers of $\cE$, which in turn implies
\begin{align*}
\mathscr{L}
&=\frac{(k+2)!}{k!}\bm L^{-+i_1\cdots i_{k}}{}_{i_1\cdots i_{k}}(\star\bm e^n)
\end{align*}
with
\begin{align*}
\frac{1}{k!}\bm L^{-+i_1\cdots i_{k+2}}{}_{j_1\cdots j_k}\bm e^{j_1}\wedge\cdots \wedge\bm e^{j_k} = \bm L^{I_1\cdots I_{k+2}}\bm Y_{I_1}\bm X_{I_2}Z_{I_3}{}^{i_1}\cdots Z_{I_{k+2}}{}^{i_k}.
\end{align*}
With this in mind, it becomes easy to construct Lagrangian densities describing important quantities in conformal and pseudo-Riemannian geometry. For example, given a metric $g=\bm\sigma^{-2}\bm\theta_i\bm\theta^i$, by requiring $D$ to be a compatible normal tractor connection, $I^I=T(\bm\sigma)$ a scale tractor, and $\bm E^I=D^{\tau_{\bm\sigma}}\bm X^I$ the corresponding generalized conformal frame, we can recover algebraic curvature invariants of $g$ in terms of products of the curvature $F_D$, the covariant derivative $DI$ and the tractor norm $|I|^2$. Schematically, we have
\begin{align*}
\bm\sigma^{2k-n}\bm X\;\bm Y\; F_D{}^k \wedge (\star \bm E^{n-2k})
&=\Weyl^k dv,
\cr
\bm\sigma^{k-n}\bm X\;\bm Y\;(DI)^k\wedge (\star \bm E^{n-k})
&=(\Ric-\tfrac{1}{n}R g)^k dv,
\cr
\bm\sigma^{-n}\bm X\;\bm Y\; |I|^{2k} (\star \bm E^{n})
&=R^{k}dv.
\end{align*}
Powers of the full Riemann tensor can also be recovered as
\begin{align*}
\bm\sigma^{2k-n}\bm X\;\bm Y\; \Big(F_D+2\bm\sigma^{-1}DI\wedge\bm E -\bm\sigma^{-2}|I|^2\bm E\wedge\bm E\Big)^k\wedge (\star \bm E^{n-2k})
&=\Riem^k dv.
\end{align*}
It should also be noted that the Lagrangian densities defined by the Pontryagin form in 4-dimensions and by the Chern-Simons form in 3-dimensions are also recovered, respectively, as the Pontryagin and Chern-Simons forms of the tractor connection.

\subsection{Simple examples in low dimensions}
\mbox{}

In low dimensions there are a few natural invariant functionals that one can consider. We briefly review these and refer to the literature for more details.

On three dimensional manifolds, one can consider the Chern-Simons functional
\begin{align*}
S_{CS}[A] &= \frac{1}{4}\int A^I{}_J \wedge dA^J{}_I + \frac{2}{3}A^I{}_J \wedge A^J{}_K \wedge A^K{}_I.
\end{align*}
Its critical points correspond to conformally flat 3D conformal metrics. See also \cite{horne_conformal_1989}.
 
In four dimension, there are several natural functionals to consider. The Pontryagin functional
 \begin{align*}
S_{P}[A]&= \frac{1}{4}\int F^I{}_J \wedge F^J{}_I
 \end{align*}
 is topological.
 
 In this dimension the Yang-Mills functional
 \begin{align*}
S_{YM}[A] &= \frac{1}{4}\int F^I{}_J \wedge \star F^J{}_I
\end{align*}
 is invariant under $P$-gauge transformation because the hodge duality is --- here the hodge dual is taken with respect to the conformal metric defined by the tractor connection. Critical points of this functional are generically difficult to interpret. However, as was shown in \cite{merkulov_supertwistor_1985}, if one restricts the variational principle to connection which are \emph{normal} then critical points correspond to Bach-free conformal metrics, which are the field equations of conformal gravity.

Finally, one can also consider
 \begin{align*}
S\left[A,I\right] &= \frac{1}{4}\int\bm\sigma^{-1}  \bm X^I I^J F^{KL}\wedge F^{MN} \epsilon_{IJKLMN}
 \end{align*}
 which, when restricted to normal connection, is also a functional for conformal gravity, see \cite{kaku_gauge_1977}.

\subsection{First order functionals for general relativity}\label{Subsec: First order functionals for general relativity}
\mbox{}

We now turn to the description of a first order tractor functional whose critical points correspond with the solutions of the (full) tractor Einstein equation.

Starting with a non-degenerate tractor connection $A^I{}_J$, a non-degenerate tractor section $I^I$ and a generalized conformal frame $\bm E^I$, consider
\begin{align}\label{Tractor functionals: tractor_functional1}
S[A,I,\bm E]=\int_M\bm L^{IJKL}\wedge(\star \bm E^{n-2})_{IJKL},
\end{align}
with
\begin{align*}
\bm L^{IJKL}=\bm\sigma^{1-n}\bm X^I\Big[I^J\Big(\tfrac{1}{2}F_A{}^{KL}+\alpha\bm\sigma^{-2}\bm E^K\wedge\bm E^L\Big) + \bm Y^JD\Big(I^K-|I|^2\bm\sigma^{-1}\bm X^K\Big)\wedge\bm E^L\Big]
\end{align*}
Here, $\bm\sigma$ is the scale induced by $I^I$, $\bm Y^I$ is the tractor frame induced by $\bm E^I$, and $\alpha=\alpha(|I|^2)$ is a quadratic polynomial of the tractor norm, given by
\begin{align*}
\alpha=(|I|^2+\tilde\Lambda)^2+\tfrac{1}{2}|I|^2-\tfrac{n-2}{2n}\tilde\Lambda,
& &
\tilde\Lambda\in\R.
\end{align*}

\begin{theorem}
For $n\geq 3$, a triple $(A^I{}_J,I^I,\bm E^I)$ is a critical points of \eqref{Tractor functionals: tractor_functional1} if and only if $A^I{}_J$ is a normal tractor connection, $I^I$ is a covariantly constant tractor with norm $|I|^2=-\tilde\Lambda$, and $\bm E^I$ is the corresponding generalized conformal frame.
\end{theorem}

\begin{proof}
This follows from a straightforward, if somewhat cumbersome, computation of the first variation of $S[A,I,\bm E]$.

One starts with the variation with respect to the tractor connection $A^I{}_J$. Here it is convenient to decompose the connection with respect to the generalized conformal frame $\bm E^I$ and consider the variation of each component separately. It should be emphasized that this can be done while maintaining manifest $P$-gauge invariance. We obtain the following equations
\begin{align*}
\bm E_I\Big(D(\bm\sigma^{-1}\bm X^I)-\bm\sigma^{-1}\bm E^I\Big)=0,
& &
\bm E_II^I=0,
\end{align*}
which are algebraic constraints on the generalized conformal frame $\bm E^I$, and
\begin{align*}
\bm E_IDI^I=0,
& &
\bm E^I\Big(d^A(\bm\sigma^{-1}\bm E_I)-\tfrac{1}{2(n-2)}\bm\sigma^{-2}\bm X_JDI^J\wedge \bm E_I\Big)=0,
\end{align*}
which are first order differential equations for $I^I$ and $A^I{}_J$.

Considering the variation with respect to the generalized conformal frame $\bm E^I$, we obtain another pair of algebraic constraints, now on $I^I$ and $A^I{}_J$,
\begin{align*}
\bm X_IDI^I=0,
& &
\Ric\Big(F_A+\tfrac{2n\alpha-2(n-1)|I|^2}{n-2}\bm\sigma^{-2}\bm E\wedge\bm E\Big)=0.
\end{align*}
Here, we have made use of the previous equations of motion.

Finally, taking the variation with respect to the tractor section $I^I$, we find a final constraint
\begin{align*}
|I|^2=-\tilde\Lambda.
\end{align*}
No other equations are imposed (in other words, some of the equations of motion obtained from the variation with respect to $I^I$ are already implied by the previous equations of motion). This is a consequence of  the $P$-gauge invariance of \eqref{Tractor functionals: tractor_functional1}.

Putting all of the equations together, we obtain
\begin{align*}
F_A{}^I{}_J\bm X^J=0 
& &
\Ric(F_A)=0,
\end{align*}
which are the normality conditions for $A^I{}_J$ (see subsection \eqref{Subsubsec: Tractor connections})
\begin{align*}
DI^I=0,
& &
|I|^2=-\tilde\Lambda,
\end{align*}
which are Einstein equations in their tractor form (see subsection \ref{Subsubsec: Tractor-Einstein equation}) and
\begin{align*}
\bm E^I=\bm\sigma D(\bm\sigma^{-1}\bm X^I),
\end{align*}
which determines $\bm E^I$ as the generalized conformal frame associated to $\bm\sigma$ (see discussion at the end of subsection \ref{Subsubsec: Tractor connections}).
\end{proof}

To relate \eqref{Tractor functionals: tractor_functional1} with more familiar constructions, let us provide here its description in components. Fixing a reference tractor decomposition we can write
\begin{align}\label{Tractor functionals: tractor_components}
	D=\left(\begin{matrix}\nabla^\tau & -\bm\theta_j & 0 \\ -\bm\xi^i & \nabla^\omega & \bm\theta^i \\ 0 & \bm\xi_j & \nabla^{-\tau} \end{matrix}\right),
	& &
	E^I = \begin{pmatrix} 0 \\ \bm e^i\\ -\bm z_j \bm e^j \end{pmatrix},
	& & 
	I^I = \begin{pmatrix} \bm\sigma \\  \mu^i \\ -\tfrac{|\mu|^2-\kappa}{2\bm\sigma} \end{pmatrix},
\end{align}
and the functional \eqref{Tractor functionals: tractor_functional1} can be then be rewritten in terms of tractor components as
\begin{align}\label{Tractor functionals: tractor_functional1 decomposition}
S=\int_M \bm\sigma^{2-n} L^{ij}\wedge (\star \bm e^{n-2})_{ij},
\end{align}
with
\begin{multline*}
L^{ij}
=\tfrac{1}{2}F_\omega{}^{ij}
+Q(\bm\sigma^{-1}\mu^i)\wedge\bm e^j+\bm\sigma^{-2}\Big(\alpha\bm e^i
-\tfrac{1}{2}\kappa\bm\theta^i\Big)\wedge\bm e^j
\cr
+\bm\sigma^{-1}\mu^id^{\tau,\omega}\bm\theta^j-\bm\xi^i\wedge(\bm e^j-\bm\theta^j)-(\bm z^i-\bm\sigma^{-1}\mu^i)(\bm\Upsilon^k-\bm\sigma^{-1}\mu^k)\bm\theta_k\wedge\bm e^j,
\end{multline*}
\begin{align*}
\bm\Upsilon_i\bm\theta^i:=\bm\sigma^{-1}d^\tau\bm\sigma,
& &
Q(\bm\sigma^{-1}\mu^i)&:=\nabla^{\tau,\omega}(\bm\sigma^{-1}\mu_i)+\bm\sigma^{-2}\Big(\mu_i\mu_j-\frac{1}{2}|\mu|^2\eta_{ij}\Big)\bm\theta^j.
\end{align*}
From this perspective, what is going on is quite transparent: The components $\bm\xi^i$ and $\tau$ of the tractor connection, the components $\bm z^i$ and $\bm e^i$ of the generalized conformal frame as well as the component $\kappa$ of the tractor section are Lagrange multipliers altogether imposing the constraints
\begin{align}\label{Tractor functionals: constraints1}
\bm e^i=\bm\theta^i,
& &
\bm\sigma^{-1}\mu^i=\bm z^i=\bm\Upsilon^i,
& &
\kappa=-\tilde\Lambda,
\end{align}
\begin{align}\label{Tractor functionals: constraints2}
\bm\xi_{ij}=-\frac{1}{(n-2)}\Big(\bm F_\omega{}^{k}{}_{ikj}-\tfrac{1}{2(n-1)}\bm F_\omega{}^{kl}{}_{kl} \eta_{ij}\Big).
\end{align}
Note that the first three constraints \eqref{Tractor functionals: constraints1} are invariantly written as 
\begin{align*}
	\bm E^I=\bm\sigma D(\bm\sigma^{-1}\bm X^I),
	& & 
	I_I\bm E^I=0
	& & 
	|I|^2=-\tilde\Lambda.
\end{align*}
and can be solved for components of the tractor section $I^I$ and the generalized conformal frame $\bm E^I$.

The variations with respect to $\omega^i{}_j$ and $\bm\theta^i$ then lead to
\begin{align*}
d^{\tau,\omega}\bm\theta^i=0,
& & 
Q(\bm\Upsilon_i)-\bm\xi_i-\tfrac{\tilde\Lambda}{2}\bm\sigma^{-2}\bm\theta_i=0,
\end{align*}
which, together with the remaining constraint \eqref{Tractor functionals: constraints2}, correspond to equations for the normality of the connection $A^I{}_J$  (see equations \eqref{normality_conditions}) and equations for $I^I$ to be a covariantly constant tractor (see equations \eqref{Tractor Einstein equations 2}).

The variation with respect to $\mu^i$ and $\bm\sigma$ do not produce new equations of motion due to gauge symmetry.

\subsection{Gauge equivalent functionals}
\mbox{}

We now consider other tractor functionals, gauge equivalent to \eqref{Tractor functionals: tractor_functional1}, which will bring us back to the Cartan-Palatini-Weyl functional \eqref{Introduction Cartan-Palatini-Weyl} described in the introduction. Here, by \emph{gauge equivalent} we mean that the corresponding spaces of critical points coincide modulo gauge transformations and that the functional coincides when evaluated on critical points.

Consider the functional
\begin{align}\label{Tractor functionals: tractor_functional2}
S_2[A,I,\bm E]=\int_M\bm L_2{}^{IJKL}\wedge (\star \bm E^{n-2})_{IJKL},
\end{align}
with
\begin{align*}
\bm L_2{}^{IJKL}=\bm\sigma^{1-n}\bm X^II^J\Big[\tfrac{1}{2}F_A{}^{KL}+\alpha\bm\sigma^{-2}\bm E^K\wedge\bm E^L + \bm\sigma^{-1} D\Big(I^K-|I|^2\bm\sigma^{-1}\bm X^K\Big)\wedge\bm E^L\Big].
\end{align*}
This is obtained from \eqref{Tractor functionals: tractor_functional1} by implementing the constraint $\bm E_I I^I=0$. Note that this constraint implies
\begin{align*}
\bm Y^I=\bm\sigma^{-1}\Big(I^I-\tfrac{1}{2}|I|^2\bm\sigma^{-1}\bm X^I\Big),
\end{align*}
and that performing the substitution in \eqref{Tractor functionals: tractor_functional1} leads immediately to \eqref{Tractor functionals: tractor_functional2}.

It must be emphasized that the implementation of constraints does not lead in general to gauge equivalent functionals. Rather, the process usually leads to a smaller number of equations (relative to the number of fields) for the critical points, thus effectively enlarging the space of solutions. And, in fact, the condition $\bm X_I DI^I=0$ does not appear as an equation for the critical points of \eqref{Tractor functionals: tractor_functional2}. On the other hand, note that we can think of \eqref{Tractor functionals: tractor_functional2} as a functional for the independent fields $A^I{}_J$, $I^I$ and $\bm E^I$ i.e. \emph{without} requiring the implicit constraint $\bm E_I I^I=0$. From this perspective, it would naively seems that critical points of \eqref{Tractor functionals: tractor_functional2} are missing both constraints $\bm X_I DI^I=0$ and $\bm E_I I^I =0$ as compared to the critical points of \eqref{Tractor functionals: tractor_functional1}.

We claim, however, that the critical points of \eqref{Tractor functionals: tractor_functional1} and \eqref{Tractor functionals: tractor_functional2} are indeed equivalent modulo gauge transformations, once appropriately understood. In fact, as compared to \eqref{Tractor functionals: tractor_functional1}, the functional \eqref{Tractor functionals: tractor_functional2} \emph{enjoys extra gauge invariance} compensating for the missing field equations: Besides the manifest invariance under $P$-gauge transformations, \eqref{Tractor functionals: tractor_functional2} is also invariant under affine transformations of the form
\begin{align}\label{Tractor functionals: tractor_functional2, gauge invariance1}
A^I{}_J\mapsto A^I{}_J+ \bm a (\bm X^I I_J- I^I\bm X_J),
& &
\bm E^I \mapsto \bm E^I +  b \bm X^I
\end{align}
with $\bm a\in\Gamma(\cL^{-1} \otimes T^*M)$ and $b \in\Gamma(T^*M)$. In particular, the missing equations $\bm X_I DI^I=0$ and $E_I I^I =0$ can be achieved as gauge fixing conditions for \eqref{Tractor functionals: tractor_functional2, gauge invariance1}. The space of critical points of \eqref{Tractor functionals: tractor_functional1}, up to $P$-gauge transformations, thus agrees with the space of critical points of \eqref{Tractor functionals: tractor_functional2}, up to $P$-gauge and affine transformations.

The origin of the additional invariance also becomes clearer in terms of tractor components \eqref{Tractor functionals: tractor_components}. The affine transformations \eqref{Tractor functionals: tractor_functional2, gauge invariance1} then read
\begin{align}\label{Tractor functionals: tractor_functional2, gauge invariance1 decomposition}
\bm\xi_k\mapsto \bm\xi_k + \bm a \; \mu_k,
& &
\tau\mapsto \tau- \bm a\bm\sigma ,
& &
\bm z_i \bm e^i &\mapsto \bm z_i \bm e^i -  b,
\end{align}
 while the functional \eqref{Tractor functionals: tractor_functional2} can be rewritten as
\begin{align}\label{Tractor functionals: tractor_functional2 decomposition}
S_2=\int_M \bm\sigma^{2-n} L_2{}^{ij}\wedge (\star \bm e^{n-2})_{ij},
\end{align}
with
\begin{align*}
L_2{}^{ij}
=\tfrac{1}{2}F_\omega{}^{ij}
+Q(\bm\sigma^{-1}\mu^i)\wedge\bm e^j+\bm\sigma^{-2}\Big(\alpha\bm e^i
-\tfrac{1}{2}\kappa\bm\theta^i\Big)\wedge\bm e^j
+\bm\sigma^{-1}\mu^id^{\tau,\omega}\bm\theta^j-\bm\xi^i\wedge(\bm e^j-\bm\theta^j).
\end{align*}
The equivalence between \eqref{Tractor functionals: tractor_functional1 decomposition} and \eqref{Tractor functionals: tractor_functional2 decomposition}is then seen directly as follows: Varying $\tau$ and $\bm z_i$ in \eqref{Tractor functionals: tractor_functional1 decomposition} we obtain a pair of algebraic constraints that can respectively be solved for $\bm z_i$ and $\tau$. The functional \eqref{Tractor functionals: tractor_functional1 decomposition} and \eqref{Tractor functionals: tractor_functional2 decomposition} are identical when evaluated on these constraints. The interpretation is however different in each case: for the first functional the constraints are implemented by the Lagrange multipliers, whereas for the second functional they are considered to be gauge fixing conditions for the extra invariance under \eqref{Tractor functionals: tractor_functional2, gauge invariance1 decomposition}.

A similar procedure can be applied once more to further implement the constraint
\begin{align*}
\bm E_I(\bm E^I-\bm\sigma D(\bm\sigma^{-1}\bm X^I))=0.
\end{align*} 
This is solved as $\bm E^I = D^{\tau-z} \bm X^I$, with $\nabla^{\tau-z}$ an arbitrary Weyl connection on $\cL$, and leads to another gauge equivalent functional:
\begin{align}\label{Tractor functionals: tractor_functional3}
S_3[A,I]=\int_M\bm L_3{}^{IJKL}\wedge (\star \bm E^{n-2})_{IJKL},
\end{align}
with
\begin{align*}
\bm L_3{}^{IJKL}=\bm\sigma^{1-n}\bm X^II^J\Big[\tfrac{1}{2}F_A{}^{KL}+\bm\sigma^{-1}DI^K\wedge\bm E^L
+\beta\bm\sigma^{-2}\bm E^K\wedge\bm E^L\Big],
\end{align*}
$\bm E^I= D^{\tau-z} \bm X^I$ and $\beta=(|I|^2+\tilde\Lambda)^2-\tfrac{1}{2}|I|^2-\tfrac{n-2}{2n}\tilde\Lambda$. 

The critical points of \eqref{Tractor functionals: tractor_functional3} are even more under-determined than were the critical points of \eqref{Tractor functionals: tractor_functional2}: as compared to those of \eqref{Tractor functionals: tractor_functional1}, they are neither required to satisfy $\bm E_I I^I=0$, $\bm X_IDI^I=0$ nor $\Ric(F_A)=0$. On the other hand, \eqref{Tractor functionals: tractor_functional3} also enjoys greater invariance than both \eqref{Tractor functionals: tractor_functional1} and \eqref{Tractor functionals: tractor_functional2}: Besides invariance under $P$-gauge transformations, it possess a larger invariance under affine transformations
\begin{align*}
A^I{}_J\mapsto A^I{}_J+\bm a(\bm X^I I_J-I^I\bm X_J)+\bm c^k(\bm X^IZ_J{}_k-Z^I{}_k\bm X_J)
& &
\bm E^I\mapsto \bm E^I+b\bm X^I
\end{align*}
for $\bm a\in\Gamma(T^*M\otimes\cL)$, $b\in\Gamma(T^*M)$ and $\bm c^k\in\Gamma(T^*M \otimes\cL^{-1}\otimes \cE)$. Once again, the missing equations can be achieved as gauge fixing conditions, thus proving the gauge equivalence of the functionals.

In terms of a tractor decomposition, this larger group of affine transformations above is generated by
\begin{align*}
\bm\xi_k\mapsto \bm\xi_k+\bm a \bm\sigma\bm\Upsilon_k +\bm c_k,
& & 
\tau\mapsto \tau-\bm a\bm\sigma,
& &
\bm z_i \bm e^i &\mapsto \bm z_i \bm e^i -  b,
\end{align*} 
and indeed the functional \eqref{Tractor functionals: tractor_functional3} can be rewritten as
\begin{align}\label{Tractor functionals: tractor_functional3 decomposition}
S_3=\int_M \bm\sigma^{2-n} L_3{}^{ij}\wedge (\star \bm\theta^{n-2})_{ij},
\end{align}
\begin{align*}
L_3{}^{ij}
=\tfrac{1}{2}F_\omega{}^{ij}+\bm\sigma^{-1}\mu^id^{\tau,\omega}\bm\theta^j
+Q(\bm\sigma^{-1}\mu^i)\wedge\bm\theta^j+\bm\sigma^{-2}\gamma\bm\theta^i\wedge\bm\theta^j,
\end{align*}
with $\gamma=(\kappa+\tilde\Lambda)^2 -\tfrac{n-2}{2n}\tilde\Lambda$. Note that this version of the functional is explicitly independent of the components $\bm\xi_i$ and $\tau$ of the tractor connection. The equivalence between the functionals \eqref{Tractor functionals: tractor_functional2 decomposition} and \eqref{Tractor functionals: tractor_functional3 decomposition} is obtained directly by varying $\bm\theta^i$ and $\bm\xi^i$ in \eqref{Tractor functionals: tractor_functional2 decomposition} and proceeds from the same type of constraint/gauge trade-off that was encountered many times in the discussions above.

We can now finally make contact with the Cartan-Palatini-Weyl functional \eqref{Introduction Cartan-Palatini-Weyl} described in the introduction. This is reached from \eqref{Tractor functionals: tractor_functional3 decomposition} by implementing the constraint $\kappa = |I|^2=-\tilde\Lambda$ (obtained by varying $\kappa$) and a partial gauge fixing condition $\bm\sigma^{-1}\mu^i=\bm\Upsilon^i$, thus reducing the $P$-gauge invariance to $W$-gauge invariance. Note that this last constraint is implemented by and can be solved for the \emph{same} field $\kappa$: its implementation thus automatically yields a gauge equivalent functional. This differs fundamentally from the previous discussion where the constraints could not be directly solved in terms of the Lagrange multipliers imposing it.

\begin{theorem}
The tractor functionals \eqref{Tractor functionals: tractor_functional1}, \eqref{Tractor functionals: tractor_functional2}, \eqref{Tractor functionals: tractor_functional3} and the Cartan-Palatini-Weyl functional \eqref{Introduction Cartan-Palatini-Weyl} define the same moduli space of critical points. What is more, the value of all these functionals coincide when evaluated on these critical points.
\end{theorem}

\subsection*{Acknowledgments}
This work was supported in part by the National Research Foundation of Korea (NRF) grant funded by the Korea government (MSIT) (2019R1F1A1060827).

CS was partially supported by a KIAS Individual Grant (SP036102) via the Center for Mathematical Challenges at Korea Institute for Advanced Study.

YH was supported by a ``Chargé de recherche'' grant from FNRS.

\bibliographystyle{ieeetr}

\end{document}

%% file: macros.tex
\newtheorem{theorem}{Theorem}[section]
\newtheorem{proposition}[theorem]{Proposition}

\newtheorem{corollary}[theorem]{Corollary}

\theoremstyle{definition}

\theoremstyle{remark}

\def\utilde#1{\mathord{\vtop{\ialign{##\crcr
$\hfil\displaystyle{#1}\hfil$\crcr\noalign{\kern1.5pt\nointerlineskip}
$\hfil\widetilde{}\hfil$\crcr\noalign{\kern1.5pt}}}}}

\newcommand{\R}{\mathbb{R}}

\newcommand{\cE}{\mathcal{E}}

\newcommand{\cJ}{\mathcal{J}}
\newcommand{\cL}{\mathcal{L}}

\newcommand{\cN}{\mathcal{N}}

\newcommand{\cP}{\mathcal{P}}

\newcommand{\cT}{\mathcal{T}}

\newcommand{\cW}{\mathcal{W}}

\newcommand{\so}{\mathfrak{so}}

\newcommand{\gt}{\tilde{g}}
\newcommand{\et}{\tilde{e}}
\newcommand{\Mt}{\widetilde{M}}

\makeatletter
\newcommand\RedeclareMathOperator{%
  \@ifstar{\def\rmo@s{m}\rmo@redeclare}{\def\rmo@s{o}\rmo@redeclare}%
}
\newcommand\rmo@redeclare[2]{%
  \begingroup \escapechar\m@ne\xdef\@gtempa{{\string#1}}\endgroup
  \expandafter\@ifundefined\@gtempa
     {\@latex@error{\noexpand#1undefined}\@ehc}%
     \relax
  \expandafter\rmo@declmathop\rmo@s{#1}{#2}}
\newcommand\rmo@declmathop[3]{%
  \DeclareRobustCommand{#2}{\qopname\newmcodes@#1{#3}}%
}
\@onlypreamble\RedeclareMathOperator
\makeatother

\RedeclareMathOperator{\Re}{Re}
\RedeclareMathOperator{\Im}{Im}

\DeclareMathOperator{\End}{End}

\DeclareMathOperator{\ISO}{ISO}
\DeclareMathOperator{\GL}{GL}

\DeclareMathOperator{\SO}{SO}
\DeclareMathOperator{\CO}{CO}

\RedeclareMathOperator{\O}{O}

\DeclareMathOperator{\Ric}{Ric}
\DeclareMathOperator{\Weyl}{Weyl}

\DeclareMathOperator{\Riem}{Riem}

\DeclareFontFamily{U}{MnSymbolC}{}
\DeclareSymbolFont{MnSyC}{U}{MnSymbolC}{m}{n}
\DeclareFontShape{U}{MnSymbolC}{m}{n}{
    <-6>  MnSymbolC5
   <6-7>  MnSymbolC6
   <7-8>  MnSymbolC7
   <8-9>  MnSymbolC8
   <9-10> MnSymbolC9
  <10-12> MnSymbolC10
  <12->   MnSymbolC12}{}
\DeclareMathSymbol{\intprod}{\mathbin}{MnSyC}{'270}

\makeatletter
\newcommand*{\boxwedge}{%
  \mathbin{%
    \mathpalette\@boxwedge{}%
  }%
}
\newcommand*{\@boxwedge}[2]{%
  \sbox0{$#1\boxplus\m@th$}%
  \dimen2=.5\dimexpr\wd0-\ht0-\dp0\relax 
  \dimen@=\dimexpr\ht0+\dp0\relax
  \def\lw{.06}
  \kern\dimen2 
  \tikz[
    line width=\lw\dimen@,
    line join=round,
    x=\dimen@,
    y=\dimen@,
  ]
  \draw
    (\lw/2,0) rectangle (1-\lw,1-\lw)
    (\lw,0) -- (.5,1-\lw-\lw/2) -- (1-\lw-\lw/2 ,0)
  ;%
  \kern\dimen2 
}
\makeatletter

\makeatletter
\newcommand*{\owedge}{%
  \mathbin{%
    \mathpalette\@owedge{}%
  }%
}
\newcommand*{\@owedge}[2]{%
  \sbox0{$#1\oplus\m@th$}%
  \dimen2=.5\dimexpr\wd0-\ht0-\dp0\relax 
  \dimen@=\dimexpr\ht0+\dp0\relax
  \def\lw{.04}
  \def\radius{.5-\lw/2}%
  \kern\dimen2 
  \tikz[
    line width=\lw\dimen@,
    line join=round,
    x=\dimen@,
    y=\dimen@,
    baseline=\dimexpr-.5\dimen@+\dp0\relax,
  ]
  \draw
    (0,0) circle[radius=\radius]
    (225:\radius) -- (0,.5-\lw) -- (-45:\radius)
  ;%
  \kern\dimen2 
}
\makeatother

\makeatletter
\let\save@mathaccent\mathaccent
\newcommand*\if@single[3]{%
  \setbox0\hbox{${\mathaccent"0362{#1}}^H$}%
  \setbox2\hbox{${\mathaccent"0362{\kern0pt#1}}^H$}%
  \ifdim\ht0=\ht2 #3\else #2\fi
  }
\newcommand*\rel@kern[1]{\kern#1\dimexpr\macc@kerna}
\newcommand*\widebar[1]{\@ifnextchar^{{\wide@bar{#1}{0}}}{\wide@bar{#1}{1}}}
\newcommand*\wide@bar[2]{\if@single{#1}{\wide@bar@{#1}{#2}{1}}{\wide@bar@{#1}{#2}{2}}}
\newcommand*\wide@bar@[3]{%
  \begingroup
  \def\mathaccent##1##2{%
    \let\mathaccent\save@mathaccent
    \if#32 \let\macc@nucleus\first@char \fi
    \setbox\z@\hbox{$\macc@style{\macc@nucleus}_{}$}%
    \setbox\tw@\hbox{$\macc@style{\macc@nucleus}{}_{}$}%
    \dimen@\wd\tw@
    \advance\dimen@-\wd\z@
    \divide\dimen@ 3
    \@tempdima\wd\tw@
    \advance\@tempdima-\scriptspace
    \divide\@tempdima 10
    \advance\dimen@-\@tempdima
    \ifdim\dimen@>\z@ \dimen@0pt\fi
    \rel@kern{0.6}\kern-\dimen@
    \if#31
      \overline{\rel@kern{-0.6}\kern\dimen@\macc@nucleus\rel@kern{0.4}\kern\dimen@}%
      \advance\dimen@0.4\dimexpr\macc@kerna
      \let\final@kern#2%
      \ifdim\dimen@<\z@ \let\final@kern1\fi
      \if\final@kern1 \kern-\dimen@\fi
    \else
      \overline{\rel@kern{-0.6}\kern\dimen@#1}%
    \fi
  }%
  \macc@depth\@ne
  \let\math@bgroup\@empty \let\math@egroup\macc@set@skewchar
  \mathsurround\z@ \frozen@everymath{\mathgroup\macc@group\relax}%
  \macc@set@skewchar\relax
  \let\mathaccentV\macc@nested@a
  \if#31
    \macc@nested@a\relax111{#1}%
  \else
    \def\gobble@till@marker##1\endmarker{}%
    \futurelet\first@char\gobble@till@marker#1\endmarker
    \ifcat\noexpand\first@char A\else
      \def\first@char{}%
    \fi
    \macc@nested@a\relax111{\first@char}%
  \fi
  \endgroup
}